\documentclass{article}
\usepackage{natbib}
\usepackage{amsmath, amsfonts, cancel}
\usepackage{amsthm}
\usepackage{multicol, geometry, dirtytalk}
\usepackage{setspace}
\usepackage{graphicx}
\usepackage{adjustbox}
\usepackage{epsfig,pict2e}
\usepackage{fancyhdr}
\usepackage[tc]{titlepic}
\usepackage{amsmath,amsfonts,amssymb,relsize,geometry}
\usepackage[utf8]{inputenc}
\usepackage[english]{babel}
\usepackage{subfiles}
\usepackage[linesnumbered,ruled]{algorithm2e}
\usepackage[]{qcircuit}
\usepackage[hang,small,bf]{caption}    
\usepackage{mathabx}
\usepackage{tikz-qtree}
\usepackage{float}
\usepackage{dsfont}
\usepackage{titlepic}

\usepackage{tikz}	
\usetikzlibrary{backgrounds,fit,decorations.pathreplacing}
\usetikzlibrary{trees}
\makeatletter
\def\BState{\State\hskip-\ALG@thistlm}
\makeatother

\usepackage{todonotes}
\theoremstyle{definition}
\newtheorem{exmp}{Example}[section]

\newcommand*{\boxednumber}[1]{%
    \expandafter\readdigit\the\numexpr#1\relax\relax
}
\newcommand*{\readdigit}[1]{%
    \ifx\relax#1\else
        \boxeddigit{#1}%
        \expandafter\readdigit
    \fi
}
\newcommand*{\boxeddigit}[1]{\fbox{#1}\hspace{-\fboxrule}}

\fancyheadoffset{1in}
\newcommand{\ket}[1]{\ensuremath{\left|#1\right\rangle}}

\newtheorem{theorem}{Theorem}[section]
\newtheorem{definition}{Definition}[section]
\newtheorem{corollary}{Corollary}[section]
\newtheorem{lemma}{Lemma}[section]

\newtheorem{conjecture}[theorem]{Conjecture}

\title{\textbf{A Gentle Introduction to \\Quantum Computing Algorithms with \\
Applications to Universal Prediction}}

\author{Elliot Catt$^1$ \and Marcus Hutter$^{1,2}$}
\date{%
     $^1$Australian National University, $^2$Deepmind\\[2ex]%
     \{elliot.carpentercatt,marcus.hutter\}@anu.edu.au\\[2ex]%
     \today \\
}

\begin{document}
 
\maketitle

\let\cleardoublepage\clearpage

\section*{Abstract}
	In this technical report we give an elementary introduction to Quantum Computing for non-physicists. In this introduction we describe in detail some of the foundational Quantum Algorithms including: the Deutsch-Jozsa Algorithm, Shor's Algorithm, Grocer Search, and Quantum Counting Algorithm and briefly the Harrow-Lloyd Algorithm. Additionally we give an introduction to Solomonoff Induction, a theoretically optimal method for prediction. We then attempt to use Quantum computing to find better algorithms for the approximation of Solomonoff Induction. This is done by using techniques from other Quantum computing algorithms to achieve a speedup in computing the speed prior, which is an approximation of Solomonoff's prior, a key part of Solomonoff Induction. The major limiting factors are that the probabilities being computed are often so small that without a sufficient (often large) amount of trials, the error may be larger than the result. If a substantial speedup in the computation of an approximation of Solomonoff Induction can be achieved through quantum computing, then this can be applied to the field of intelligent agents as a key part of an approximation of the agent AIXI.

 \newpage
 
\parskip=-0.5ex\tableofcontents

\let\cleardoublepage\clearpage
    
\listoffigures 

\listofalgorithms
 
\section{Introduction}
Quantum Computing (QC) is a form of computation that has been shown to outperform classical computing on several tasks. These tasks include factoring numbers, querying databases, and some aspects of machine learning. There has been little work done on using Quantum Computing to achieve speedups in areas relevant to the study of inductive reasoning, reinforcement learning and algorithmic information theory. 

By itself, inductive reasoning is a powerful tool that can solve a myriad of problems. One of the greatest of these, the problem of Artificial General Intelligence (AGI), is unlikely to be solved by inductive reasoning alone. \cite{Hutter:04uaibook} presented a theoretical solution to the problem of AGI with the optimal agent AIXI. The agent AIXI is a combination of reinforcement learning and algorithmic information theory. AIXI is unfortunately incomputable, however it can be approximated.

In this technical report we present an introduction to Quantum Computing from a Computer Science/Mathematics perspective avoiding physics terminology as well as a gentle introduction to Universal Induction. We conclude with a presentation of putting these two ideas together in the form of Quantum Algorithms which improve on the classical method of approximating AIXI.

 In section 2, we will give a short description of the background required for this technical report; this includes computability, probability, 
 and computational complexity theory. 

 In section 3, we will describe quantum computing, from its inception to the present; the foundation of quantum complexity theory; advances in quantum computability; and some well-known quantum algorithms which provide large speedups over their classical counterparts, these include: the Deutsch-Jozsa Algorithm, Shor's Algorithm, Grocer Search, and Quantum Counting Algorithm and briefly the Harrow-Lloyd Algorithm. 
 
 In section 4, we go over the hardness of counting, specifically the implications of a fast classical or quantum algorithm. 
 
 In section 5, we will describe the problem of induction, specifically inductive reasoning, as well as some approaches to the problem of induction. This will include an explanation of Kolmogorov complexity and Solomonoff Induction.

 In section 6, we discuss the Speed prior, and how computing the Speed prior relates to counting, and that therefore it is unlikely that there exists an exponential speedup for it with quantum computing; we then describe some of the classical algorithms which could compute the Speed prior, and go on to present two quantum algorithms to compute the Speed prior.
 
 The advantage in using the quantum algorithms to compute the Speed prior 
 is that there is a potential speedup in the time taken. The first quantum algorithm provides a quadratic speedup over the classical method, while the second provides an exponential speedup, however the approximation is more crude.
 
 By providing quantum algorithms to the problem of approximating Solomonoff induction, via the speed prior we have also provided an approach to approximate AIXI using quantum computing, which we call AIXIq.

\section{Preliminaries}
 
In this section we will cover some of the prerequisites of this paper, namely the basics of probability, computability and complexity.

\subsection{Probability}

Probability is the study of events and their likelihoods. To describe later results we will need to define probability measures on binary strings. We use the notation for binary strings: $\mathbb{B} = \{ 0,1\}$, $\mathbb{B}^* = \bigcup_n \{0,1\}^n$, $\epsilon$ is the empty string, and $\mathbb{B}^{\infty}$ is the set of all one way infinite sequences of $\mathbb{B}$.

\begin{definition}[Cylinder set \citep{ming2014kolmogorov}]
	A cylinder set $\Gamma_x \subseteq \mathbb{B}^{\infty}$, for $x\in\mathbb{B}^*$, is the set
	\[ \Gamma_x = \{ x\omega : \omega \in \mathbb{B}^{\infty} \} \]
\end{definition}

Let $\mathbb{G} = \{ \Gamma_x : x\in \mathbb{B}^* \}$ then a function $\mu' :\mathbb{G}\to \mathbb{R}$ is defines a probability measure if
\begin{align*}
		1. &\ \mu'(\Gamma_{\epsilon} ) = 1 \\
		2. &\ \mu' (\Gamma_x) = \sum_{b\in\mathbb{B}^*} \mu' (\Gamma_{xb})
	\end{align*} We will use the notation where $\mu (x) = \mu'(\Gamma_x )$, then the definition of a measure can be written as
\begin{definition}[Probability Measure \citep{ming2014kolmogorov}]
	The a function $\mu :\mathbb{B}^*\to \mathbb{R}$ is defines a probability measure if
\begin{align*}
		1. &\ \mu(\epsilon ) = 1 \\
		2. &\ \mu (x) = \sum_{b\in\mathbb{B}^*} \mu (xb)
	\end{align*}
\end{definition} 

Expected value is another tool from probability which we will be using.

\begin{definition}
	Given a probability mass function $P(x)$ the (discrete) expected value of a function $f$ is defined as follows
	\begin{equation}
		\mathbb{E}[f] = \sum_x f(x)P(x)
	\end{equation}
\end{definition}

An example of the expected value is given below,

\begin{exmp}
	The expected value of the sum of rolling two fair 6-sided dice with the probability mass function $P(x)=1/36$, the expected value formula gives us the following formula
\[ \sum_{a,b\in \{1,2,3,4,5,6\} }(a+b) \cdot \frac{1}{36} = 7 \]
\end{exmp}

\subsection{Computability}
First proposed by Alan Turing in \cite{turing1937computable}, Turing machines are a class of machines with a tape of zeros and ones, a head which moves along and writes on the tape, and a program determining where the head should go. This simple form of computation is the foundation of computer science, which as a field could be described as ``things that can be done with Turing machines''.

Turing machines are not the only form of computation, another equal form of computation is Lambda calculus (Church), which takes a much more functional (in the mathematical sense) approach to computing. Other forms of equivalent computation include (but are not limited to) partial recursive functions, register machines, and Markov algorithms. The idea that all forms of computation are ``equivalent'' is called the Church-Turing thesis.

Formally we define a Turing Machine as follows,
\begin{definition}[\cite{bernstein1997quantum}]
\label{tmdef}
	A deterministic Turing machine is a triplet $(\Sigma,Q,\delta )$, where $\Sigma$ is a finite alphabet with an identified blank symbol $\#$, $Q$ is a finite set of states with identified initial state $q_0$ and finial state $q_f\neq q_0$, and $\delta$, a deterministic transition function, is a function \begin{equation}
		\delta \ : \ Q\times \Sigma \to \Sigma \times Q \times \{ L,R \}
	\end{equation} 
	Here $\{L,R\}$ denote left and right, directions to move on the tape. The state $q_f$ is also called the Halting state.
\end{definition}
Turing machines can be thought of as a head moving along an infinite tape, and on this tape are elements from the alphabet $\{\# \}\cup\Sigma$. The head moves up and down the tape, and reads and writes according to the function $\delta$.

Each Turing machine can be represented by a partial function, $\delta^*$, which takes the initial Turing machine tape as input, and outputs the contents of the tape once the Turing machine halts, if it does halt, and outputs undefined if it does not halt. Hence being a partial function. When referring to the Turing machine as a function we will be referring to the $\delta^*$ of that Turing machine. 

A configuration of a Turing machine, which is a current description of the Turing Machine is defined as follows

\begin{definition}[Configuration]
\label{confdef}
	A configuration (of a Turing Machine) is a tuple $(d,h,q)$ where $d$ is a description of the contents of the tape, $h$ is the location of the head symbol, and $q$ represents the state the Turing machine is in.
\end{definition}
To receive an output on our computation we require the Turing machine to halt, that is, eventually enter the final state $q_f$. Unfortunately this is not always the case. For example consider the machines below,

\begin{exmp}
\begin{align*}
	\Sigma &= \{\#,1\} \\
	Q &= (q_0,q_f) \\
	\delta(q_0,\#) &= (1,q_0,R) \\
	\delta(q_0,1) &= (1,q_0,R)
\end{align*}
This machine writes 1, then moves right forever. It will never halt since the function $\delta$ never maps to the state $q_f$. Below is the state diagram of this Turing machine,
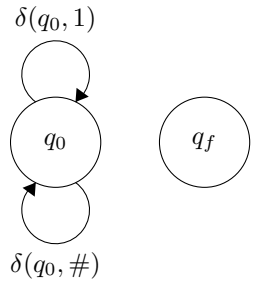
\begin{center}
\begin{tikzpicture}[scale=0.2]
\tikzstyle{every node}+=[inner sep=0pt]
\draw [black] (36.1,-33.4) circle (3);
\draw (36.1,-33.4) node {$q_0$};
\draw [black] (46,-33.4) circle (3);
\draw (46,-33.4) node {$q_f$};
\draw [black] (37.423,-36.08) arc (54:-234:2.25);
\draw (36.1,-40.65) node [below] {$\delta(q_0,\#)$};
\fill [black] (34.78,-36.08) -- (33.9,-36.43) -- (34.71,-37.02);
\draw [black] (34.777,-30.72) arc (234:-54:2.25);
\draw (36.1,-26.15) node [above] {$\delta(q_0,1)$};
\fill [black] (37.42,-30.72) -- (38.3,-30.37) -- (37.49,-29.78);
\end{tikzpicture}
\end{center}
\end{exmp}

A non-halting machine may not use an infinite amount of tape, as seen in the next example.

\begin{exmp}
\begin{align*}
	\Sigma &= \{\#,1\} \\
	Q &= (q_0,q_1,q_f) \\
	\delta(q_0,\#) &= (1,q_1,R) \\
	\delta(q_1,\#) &= (1,q_0,L) \\
	\delta(q_0,1) &= (1,q_1,R) \\
	\delta(q_1,1) &= (1,q_0,L)
\end{align*}
	This machine will move left then right and so on. Again this machine will never halt since the function $\delta$ never maps to the state $q_f$. Below is the state diagram of this Turing machine,
	
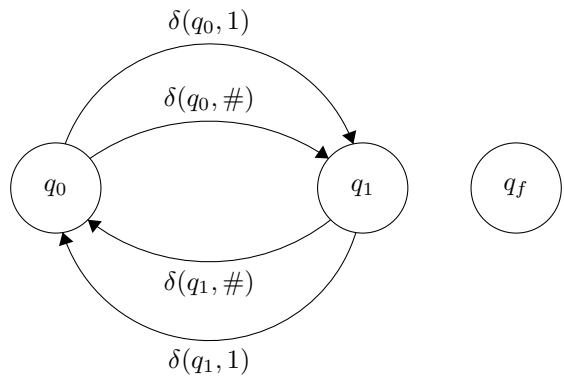
\begin{center}
\begin{tikzpicture}[scale=0.2]
\tikzstyle{every node}+=[inner sep=0pt]
\draw [black] (22.8,-34) circle (3);
\draw (22.8,-34) node {$q_0$};
\draw [black] (43.2,-34) circle (3);
\draw (43.2,-34) node {$q_1$};
\draw [black] (53.4,-34) circle (3);
\draw (53.4,-34) node {$q_f$};
\draw [black] (23.42,-31.076) arc (159.6208:20.3792:10.22);
\fill [black] (42.58,-31.08) -- (42.77,-30.15) -- (41.83,-30.5);
\draw (33,-23.91) node [above] {$\delta(q_0,1)$};
\draw [black] (25.062,-32.038) arc (124.76934:55.23066:13.92);
\fill [black] (40.94,-32.04) -- (40.57,-31.17) -- (40,-31.99);
\draw (33,-29.05) node [above] {$\delta(q_0,\#)$};
\draw [black] (42.747,-36.955) arc (-17.14503:-162.85497:10.2);
\fill [black] (23.25,-36.95) -- (23.01,-37.87) -- (23.97,-37.57);
\draw (33,-44.65) node [below] {$\delta(q_1,1)$};
\draw [black] (41.076,-36.109) arc (-51.78538:-128.21462:13.055);
\fill [black] (24.92,-36.11) -- (25.24,-37) -- (25.86,-36.21);
\draw (33,-39.41) node [below] {$\delta(q_1,\#)$};
\end{tikzpicture}
\end{center}
\end{exmp}

Turing machines can also be viewed as functions equivalent to partial recursive functions \cite{boolos2002computability}.

The `Halting problem' is determining whether or not a Turing machine will halt on a given input. Turing proved that one cannot in fact use a Turing machine to determine that another Turing machine will not halt on any input. 

\begin{theorem}[Halting problem]
	There does not exist a Turing Machine which can determine if a given Turing Machine will not halt on any input.
\end{theorem}

The `Halting problem' is a good example of something which is incomputable: it is something which no Turing machine can compute.

\subsection{Complexity Theory}
Complexity theory is the study of how long it takes to compute a function on a Turing machine. 

Specifically, given a Turing machine (which is representing a function $\delta^*$), how many steps does the Turing machine take to compute that function in terms of the size of the input of the Turing machine.

For example consider an ADD Turing machine, which takes two (unary) numbers and adds them together, as shown below

\begin{exmp}
	\begin{align*}
	\Sigma &= \{\#,1\} \\
	Q &= (q_0,q_1,q_2,q_f) \\
	\delta(q_0,1) &= (1,q_0,R) \\
	\delta(q_0,\#) &= (1,q_1,R) \\
	\delta(q_1,1) &= (1,q_1,R) \\
	\delta(q_1,\#) &= (\#,q_2,L) \\
	\delta(q_2,1) &= (\#,q_f,R)
\end{align*}
The ADD Turing machine takes an input of two unary numbers separated by a $\#$, $(1^l \# 1^m$, then returns the sum of those two numbers, $(1^{l+m})$. Below is the state diagram of this Turing machine,
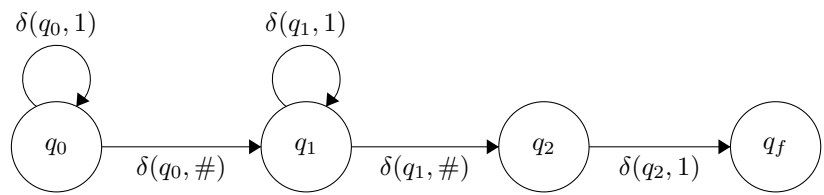
\begin{center}
\begin{tikzpicture}[scale=0.2]
\tikzstyle{every node}+=[inner sep=0pt]
\draw [black] (13.4,-29.5) circle (3);
\draw (13.4,-29.5) node {$q_0$};
\draw [black] (30,-29.5) circle (3);
\draw (30,-29.5) node {$q_1$};
\draw [black] (45.8,-29.5) circle (3);
\draw (45.8,-29.5) node {$q_2$};
\draw [black] (61.2,-29.5) circle (3);
\draw (61.2,-29.5) node {$q_f$};
\draw [black] (12.077,-26.82) arc (234:-54:2.25);
\draw (13.4,-22.25) node [above] {$\delta(q_0,1)$};
\fill [black] (14.72,-26.82) -- (15.6,-26.47) -- (14.79,-25.88);
\draw [black] (16.4,-29.5) -- (27,-29.5);
\fill [black] (27,-29.5) -- (26.2,-29) -- (26.2,-30);
\draw (21.7,-30) node [below] {$\delta(q_0,\#)$};
\draw [black] (28.677,-26.82) arc (234:-54:2.25);
\draw (30,-22.25) node [above] {$\delta(q_1,1)$};
\fill [black] (31.32,-26.82) -- (32.2,-26.47) -- (31.39,-25.88);
\draw [black] (33,-29.5) -- (42.8,-29.5);
\fill [black] (42.8,-29.5) -- (42,-29) -- (42,-30);
\draw (37.9,-30) node [below] {$\delta(q_1,\#)$};
\draw [black] (48.8,-29.5) -- (58.2,-29.5);
\fill [black] (58.2,-29.5) -- (57.4,-29) -- (57.4,-30);
\draw (53.5,-30) node [below] {$\delta(q_2,1)$};
\end{tikzpicture}
\end{center}
\end{exmp}

One can see that if the length of the numbers is $n=l+m+1$, then the number of steps required to compute this ADD function is $n+2$. In complexity theory, the big O notation is used to describe the time taken. 

\begin{definition}
	A function $f$ is $f\in O(g)$ (read as big O of $g$) if $\exists c\in\mathbb{R}$ and $n_0$ such that $|f(n)|\leq c |g(n)|$ for all $n>n_0$. We represent this by the notation $f(n)=O(g(n))$.
\end{definition}

From our previous example we would say the ADD function is $O(n)$, since there exists a $c=2$ and a $n_0=4$ such that $n+3 < 2n$ for all $n>4$. A problem is said to be solved in polynomial time if there exists a function which solves the problem and takes time $O(f(n))$, where $f$ is a polynomial of the size of the input $n$. Similarly a problem is said to be solved in exponential time if there exists a function which solves the problem and takes time $O(f(n))$, where $f$ is an exponential function of the size of the input $n$.

The two most important classes in complexity theory are $\mathbf{P}$ and $\mathbf{NP}$. The class $\mathbf{P}$ is the class of all problems that can be solved in polynomial time; the class $\mathbf{NP}$ is the class of all problems that can be verified in polynomial time. Verifying a problem is being given a solution to the problem and checking if it is correct. For example if the problem was add $4$ and $5$, then we could be given a solution, such as $10$ and we have to check if it is correct. In this case it is not.

 The most famous problem in complexity theory is whether or not $\mathbf{P}=\mathbf{NP}$. This is essentially asking if being able to verify a problem in polynomial time implies we can solve the problem in polynomial time.

The reason these classes, and the above question, are so important comes down to the fact that $\mathbf{P}$ is also the class of problems we can efficiently (and therefore physically) solve. In contrast, $\mathbf{NP}$ contains many problems which we cannot efficiently solve. Additionally several of these hard-to-solve problems are very relevant, for example the travelling salesman problem, protein folding, and RSA cryptography.

A problem $B$ can be reduced in polynomial time to another problem $A$ if there exists transformation which takes at most polynomial time and transforms every instance of a problem $B$ into an instance of the problem $A$ such that the transformation of the solution to the instance of the new problem $A$ is the solution to the instance of the problem $B$.
For example, any hamiltonian cycle problem can be reduced to a travelling salesman problem by giving all the edges of the hamiltonian cycle problem weight 1, and creating edges with weight 2 where all the non existent edges of the hamiltonian cycle problem are. The solution to the resulting travelling salesman problem can be transformed to the solution to the initial hamiltonian cycle problem, and the transformation takes at msot polynomial time.

If every problem in $\mathbf{NP}$ can be reduced in polynomial time to a problem $A$, then we say $A$ is $\mathbf{NP}$-hard. If a problem is in $\mathbf{NP}$ and is $\mathbf{NP}$-hard then we say the problem is $\mathbf{NP}$-complete. The above travelling salesman problem is an example of a problem that is $\mathbf{NP}$-complete.

\section{Quantum Computing}

In this section we will give an introduction to quantum computing. In doing so, we will discuss the foundations of quantum computing and establish the notation used. We will then move into some more advanced topics such as five key quantum algorithms, the development and progress of quantum complexity theory, some results in quantum computability, and quantum algorithmic information theory. We will assume the reader has some familiarity with Turing machines, the complexity classes \textbf{P} and \textbf{NP}, as well as the basics of Boolean circuits. We will not go into detail on the physical construction of a quantum computer; but will however mention some recent progress.

\subsection{Introduction}
The central idea in Quantum Computing (QC) is to use the quantum nature of the universe to provide speedups in computation that would otherwise be impossible in classical computing \citep{feynman1982simulating}. The physical realisation of this is called quantum supremacy \citep{aaronson2016complexity}. This quantum supremacy relies on two key features: first, a physically constructed quantum computer, and second, a Quantum algorithm running on the quantum computer that is beyond any (current) classical computation. We will be focusing on the quantum algorithm side. A complete description of the physical and algorithmic aspects of quantum computing can be found in \cite{nielsen2002quantum}.

There are two ways in which we will describe quantum computation: first, the way in which it was established, as quantum Turing machines. Second, quantum Circuits, which give a clear description of quantum algorithms. It has been shown that these are equivalent by \cite{yao1993quantum}.

\subsection{Quantum Turing Machines}
The first formal definition of a Quantum Turing Machine (QTM) is from \cite{deutsch1985quantum}. A main topic of \cite{deutsch1985quantum}, and an important aspect of the purpose of quantum computing was the Church-Turing thesis and its expanded form, which we will discuss here before we get to the exact definition of Quantum Turing Machines.

\subsubsection{Church-Turing Thesis}
The Church-Turing thesis is as follows:
\begin{quote}
	\textit{Every `function which could be regarded as computable' can be computed by a universal Turing Machine.}
\end{quote} 
This states that anything which we may wish to compute could be simulated on a universal Turing machine. However it does not mention anything about the time it would take for the universal Turing machine to simulate a computation.

Deutsch proposed a new version of the thesis, reflecting the physical nature of reality, now called the Church-Turing-Deutsch thesis.
\begin{quote}
 	\textit{Every finitely realisable physical system can be perfectly simulated by a universal model computing machine operating by finite means.}
 \end{quote} 
 Here, Deutsch is interested in a ``finitely realisable physical system'' which we want to simulate, not a ``function which could be regarded as computable''. This is quite an important distinction as we can (quite easily) construct functions that are theoretically possible to compute with a Turing machine, however the functions are not representing any finitely realisable physical system. The constraint here is physical, as the time required to compute some functions is greater than the heat-death of the universe: For example, the Ackermann function \citep{calude1979first}.
 
Lastly from Complexity theory, there is the related Extended-Church-Turing thesis by \cite{kaye2007introduction}: \begin{quote}
 	\textit{A probabilistic Turing machine can efficiently simulate any realistic model of computation.}
 \end{quote} With the inclusion of ``efficiently'', one excludes a large number of  problems which are both computable and also represent finitely realisable physical systems. It is important to note that if quantum computing became a ``realistic model of computation'', then this thesis may be negated, since probabilistic Turing machines cannot efficiently simulate quantum Turing machines.

\subsubsection{Formalising Quantum Turing Machines}
In a certain sense, a quantum Turing machine is essentially a probabilistic Turing machine which uses the $L^2$ norm instead of the $L^1$ norm; has complex-valued amplitudes in the place of non-negative real probabilities; and a complex-valued unitary transition matrix instead of a stochastic one. This is similar to how quantum mechanics, without the physics, is (essentially) probability theory with the $L^2$ norm  \citep{aaronson2013quantum}.

With a classical (deterministic) Turing machine, the actions performed and state transitions are unsurprisingly deterministic; this means what we expect to happen will exactly happen. With a classical probabilistic Turing machine, each action will occur with some (non-negative) real probability. This means that the machine could write 0 with some probability $p$ and write 1 with some other probability $q$. As mentioned above, a quantum Turing machine is quite similar, except instead of (non-negative) real probabilities, the quantum Turing machine uses complex-valued amplitudes, where the sum of the squares of the absolute value of the amplitudes must be 1 at all times.

It can help to think of this kind of computation as a tree, where each node is a configuration of the quantum Turing machine. The branch width of this tree is the number of possible actions and the depth is the number of actions taken. At each time step, we move down the tree according to the unitary transition matrix.

Below is an example of deterministic, probabilistic, and quantum Turing machine trees representing the actions of moving left and right achieving different tape configurations, with the edges representing the probability (or amplitudes) of performing that action.

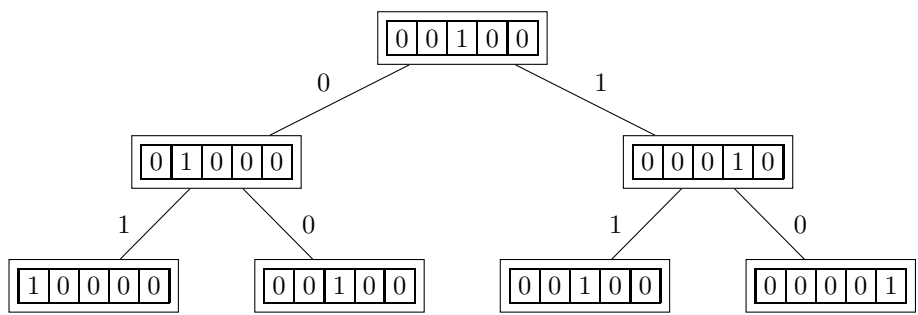
\begin{figure}[H]
\begin{tikzpicture}[every tree node/.style={draw},
   level distance=1.65cm,sibling distance=1cm,
   edge from parent path={(\tikzparentnode) -- (\tikzchildnode)}]
\Tree[.\boxednumber{0}\boxednumber{0}\boxednumber{100} \edge node[auto=right] {$0$}; [.\boxednumber{0}\boxednumber{1000} \edge node[midway,left,xshift=-2mm] {$1$};  [.\boxednumber{10000} ] \edge node[midway,right,xshift=2mm] {$0$};
               [.\boxednumber{0}\boxednumber{0}\boxednumber{100} ]] \edge node[auto=left] {$1$};
          [.\boxednumber{0}\boxednumber{0}\boxednumber{0}\boxednumber{10} \edge node[midway,left,xshift=-2mm] {$1$};  [.\boxednumber{0}\boxednumber{0}\boxednumber{100}  ] \edge node[midway,right,xshift=2mm] {$0$};
                [.\boxednumber{0}\boxednumber{0}\boxednumber{0}\boxednumber{0}\boxednumber{1} ]]]
\end{tikzpicture}
                  \caption{Deterministic Turing Machine}
\end{figure}

Then the Probabilistic Turing machine. Note that at each vertex the branching  edges must add to 1.

\begin{figure}[H]
\begin{tikzpicture}[every tree node/.style={draw},
   level distance=1.65cm,sibling distance=1cm,
   edge from parent path={(\tikzparentnode) -- (\tikzchildnode)}]
\Tree[.\boxednumber{0}\boxednumber{0}\boxednumber{100} \edge node[auto=right] {$\frac{1}{2}$}; [.\boxednumber{0}\boxednumber{1000} \edge node[midway,left,xshift=-2mm] {$\frac{5}{8}$};  [.\boxednumber{10000} ] \edge node[midway,right,xshift=2mm] {$\frac{3}{8}$};
               [.\boxednumber{0}\boxednumber{0}\boxednumber{100} ]] \edge node[auto=left] {$\frac{1}{2}$};
          [.\boxednumber{0}\boxednumber{0}\boxednumber{0}\boxednumber{10} \edge node[midway,left,xshift=-2mm] {$\frac{2}{3}$};  [.\boxednumber{0}\boxednumber{0}\boxednumber{100}  ] \edge node[midway,right,xshift=2mm] {$\frac{1}{3}$};
                [.\boxednumber{0}\boxednumber{0}\boxednumber{0}\boxednumber{0}\boxednumber{1} ]]]
\end{tikzpicture}
                  \caption{Probabilistic Turing Machine}
\end{figure}
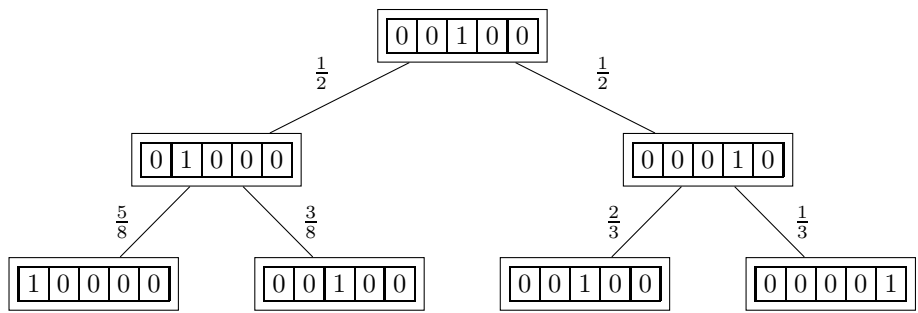

Lastly the Quantum Turing Machine. Note here that the sum of the squared absolute value at each vertex of the branches must add to 1. 

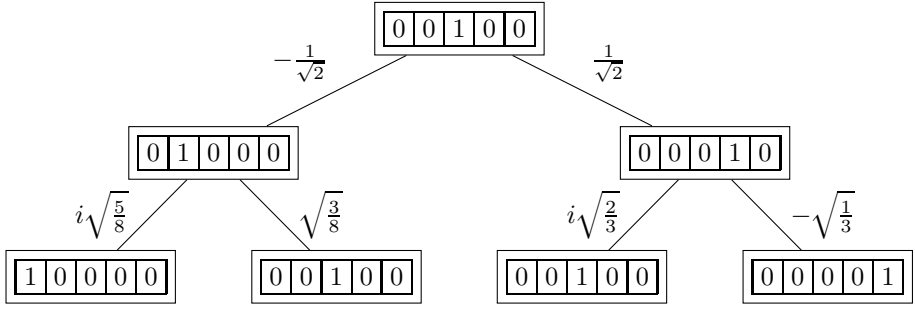
\begin{figure}[H]
\begin{tikzpicture}[every tree node/.style={draw},
   level distance=1.65cm,sibling distance=1cm,
   edge from parent path={(\tikzparentnode) -- (\tikzchildnode)}]
\Tree[.\boxednumber{0}\boxednumber{0}\boxednumber{100} \edge node[auto=right] {$-\frac{1}{\sqrt{2}}$}; [.\boxednumber{0}\boxednumber{1000} \edge node[midway,left,xshift=-2mm] {$i\sqrt{\frac{5}{8}}$};  [.\boxednumber{10000} ] \edge node[midway,right,xshift=2mm] {$\sqrt{\frac{3}{8}}$};
               [.\boxednumber{0}\boxednumber{0}\boxednumber{100} ]] \edge node[auto=left] {$\frac{1}{\sqrt{2}}$};
          [.\boxednumber{0}\boxednumber{0}\boxednumber{0}\boxednumber{10} \edge node[midway,left,xshift=-2mm] {$i\sqrt{\frac{2}{3}}$};  [.\boxednumber{0}\boxednumber{0}\boxednumber{100}  ] \edge node[midway,right,xshift=2mm] {$-\sqrt{\frac{1}{3}}$};
                [.\boxednumber{0}\boxednumber{0}\boxednumber{0}\boxednumber{0}\boxednumber{1} ]]]
\end{tikzpicture}
  \caption{Quantum Turing Machine}
\end{figure}

An important part of probabilistic Turing machines and quantum Turing machines is that two separate branches may have the same configuration. In the above examples, it is represented in the probabilistic Turing machine by the configuration \boxednumber{0}\boxednumber{0}\boxednumber{100} occurring with probability $\frac{1}{2}\cdot \frac{3}{8} +\frac{1}{2}\cdot \frac{2}{3} = \frac{19}{48}$, and similarly for the quantum Turing machine, the configuration \boxednumber{0}\boxednumber{0}\boxednumber{100} occurring with amplitude $-\frac{1}{\sqrt{2}}\cdot \sqrt{\frac{3}{8}} +\frac{1}{\sqrt{2}}\cdot i\sqrt{\frac{2}{3}} =-\sqrt{\frac{3}{16}} +i\frac{1}{\sqrt{3}} $. This addition of different paths to the same configuration is quintessential to all Quantum algorithms.

Before formally defining the quantum Turing machine, we need to define the subset of the complex numbers that define the quantum Turing machine. Let $\tilde{\mathbb{C}}$ be the set of complex numbers $\alpha \in \mathbb{C}$ for which there is a deterministic algorithm that computes the real and imaginary parts of $\alpha$ with an error of at most $2^{-n}$ in time that is a polynomial of $n$. 

Additionally, recall that a configuration in Definition \ref{confdef} is the combination of the tape, state and head position of a Turing Machine. Now we can formally define quantum Turing machines.

\begin{definition}[\cite{bernstein1997quantum}]
	A Quantum Turing Machine $M$ is defined, much like a classical Turing Machine (Definition \ref{tmdef}), by a triplet $(\Sigma ,Q,\delta )$ where $\Sigma$ is a finite alphabet with an identified blank symbol ($\#$), $Q$ is a finite set of states with identified initial state $q_0$ and final state $q_f \neq q_0$, and $\delta$, the quantum transition function, \[ \delta \ :\ Q \ \times \Sigma \ \to \ \tilde{\mathbb{C}}^{\Sigma \ \times \ Q \ \times \ \{L,R\} }. \] The QTM $M$ has a two-way infinite tape of cells indexed by $\mathbb{Z}$, each holding symbols from $\Sigma$, and a single read/write tape head that moves along the tape. A configuration or instantaneous description of the QTM is a complete description of the contents of the tape, the location of the tape head, and the state $q\in Q$ of the finite control. 
	
	Let $S$ be the inner-product space of finite complex linear combinations of configurations of $M$ with the Euclidean norm. We call each element $\phi \in S$ a superposition of $M$. 
	\[ \phi = \sum_{i}\alpha_i m_i \]
	The QTM $M$ defines a linear operator $U_M : S\to S$, called the time evolution operator of $M$, as follows: if $M$ starts in configuration $c$ with current state $q_k$, and scans symbol $a$, then after one step $M$ will be in a superposition of configurations $\psi = \sum_i \alpha_i c_i$, where each nonzero $\alpha_i$ corresponds to a transition $\delta(q_k,a,b,q_j,d)$, and $c_i$ is the new configuration (Definition \ref{confdef}) that results from applying this transform to $c$. Extending this map to the entire space $S$ through linearity gives the time evolution operator $U_M$.
\end{definition}

However instead of using Quantum Turing machines to define our Quantum Computation, we will instead mainly use Quantum Circuits.

\subsection{Quantum Circuits}
An easier way to describe some quantum computing systems is with quantum circuits. Quantum circuits were first proposed by \cite{yao1993quantum} where it was shown that they are equivalent to quantum Turing machines, in the sense that one can simulate the other with a polynomial slowdown. This equivalence was also demonstrated by \cite{nishimura2009perfect}.

Quantum circuits can be thought of as classical Boolean circuits, except instead of the classical bits which take values 0 and 1 (False and True), a quantum circuit uses qubits, where each qubit takes a value in a complex superposition of 0 and 1. We can represent this as a pair of amplitudes $(\alpha,\beta)\in\mathbb{C}^2$ which has the property $|\alpha|^2 + |\beta|^2 = 1$. Here $\alpha$ is representing the amplitude (complex probability) of being 0, and $\beta$ likewise for 1. We will use the bra-ket notation to describe qubits.
\subsection{Bra-ket notation}
The Dirac bra-ket \citep{dirac1939new} notation is as follows: first we use it to represent the standard basis vectors of $\mathbb{C}^2$
\[ \ket{0} = \begin{pmatrix}
	1 \\
	0
\end{pmatrix},\ \ket{1} = \begin{pmatrix}
	0 \\
	1
\end{pmatrix} \]

with a single qubit being described as
\[ |\phi\rangle = \alpha \ket{0} + \beta \ket{1} = \begin{pmatrix}
	\alpha \\
	\beta
\end{pmatrix}.  \] This asymmetrical notation is called a ket. This choice of notation is far less cumbersome than regular vector notation.

For $\ket{a} = \begin{pmatrix}
	\alpha_0 \\
	\alpha_1
\end{pmatrix}$ and $\ket{b} = \begin{pmatrix}
	\beta_0 \\
	\beta_1
\end{pmatrix}$,

we will also define the tensor product in bra-ket notation as follows:
\[ \ket{a} \otimes \ket{b} = \ket{a} \ket{b} = \ket{ab} = \begin{pmatrix}
	\alpha_0 \beta_0 \\
	\alpha_0 \beta_1 \\
	\alpha_1 \beta_0 \\
	\alpha_1 \beta_1 
\end{pmatrix} \]

For example, instead of writing $\ket{0} \otimes \ket{0} \otimes \ket{1}$ we will write \[ \ket{001} = \begin{pmatrix}
	1 \\
	0
\end{pmatrix} \otimes \begin{pmatrix}
	1 \\
	0
\end{pmatrix} \otimes \begin{pmatrix}
	0 \\
	1
\end{pmatrix} = \begin{pmatrix}
	1\cdot 1\cdot 0 \\
	1\cdot 1\cdot 1 \\
	1\cdot 0\cdot 0 \\
	1\cdot 0\cdot 1 \\
	0\cdot 1\cdot 0 \\
	0\cdot 1\cdot 1 \\
	0\cdot 0\cdot 0 \\
	0\cdot 0\cdot 1 
\end{pmatrix} = \begin{pmatrix}
	0 \\
	1 \\
	0 \\
	0 \\
	0 \\
	0 \\
	0 \\
	0 
\end{pmatrix}  \].

We will also raise some qubits to the power of tensors, for example: 
\[ \ket{a}^{\otimes 4} = \ket{a} \otimes \ket{a} \otimes \ket{a} \otimes \ket{a}=\begin{pmatrix}
	\alpha_0\cdot\alpha_0\cdot\alpha_0\cdot\alpha_0 \\
	\alpha_0\cdot\alpha_0\cdot\alpha_0\cdot\alpha_1 \\
	\vdots \\
	\alpha_1\cdot\alpha_1\cdot\alpha_1\cdot\alpha_0 \\
	\alpha_1\cdot\alpha_1\cdot\alpha_1\cdot\alpha_1
\end{pmatrix}  \]
Additionally we will define the conjugate transpose as
\[ \langle a| = \ket{a}^{\dag}=(\bar{\alpha_0} ,\bar{\alpha_1} ) \] where $\bar{\alpha}$ is the complex conjugate of $\alpha$. This notation is called a bra.

Using the two together we can write the inner product as
\[\langle a| \ket{b} = \langle a\ket{b} = \bar{\alpha_0} \beta_0 + \bar{\alpha_1} \beta_1  \]

and the outer product as
\[ \ket{b} \langle a|= \begin{pmatrix}
	\beta_0 \bar{\alpha_0} & \beta_0 \bar{\alpha_1} \\
	\beta_1 \bar{\alpha_0} & \beta_1 \bar{\alpha_1}
\end{pmatrix} \]

This is not the only notation that is used in Quantum computing, however this notation is the most simple and more importantly it is short.

\subsubsection{Quantum Gates}

Continuing this notation we can then represent quantum gates as matrices (transforms) which are applied to a superposition of qubits. It is important that they take a (collection of) qubits from one superposition to another. Specifically, it is required that the quantum gate conserves the $\|\alpha \|_2^2 = 1$ property. 

So a matrix that conserves the superposition is unitary, i.e., its inverse is also its conjugate transpose. This is no surprise, since the state transition matrices of quantum physics must also be unitary. It has been shown that if any linear matrix was allowed, quantum computing would be unreasonably powerful \citep{aaronson2005quantum}.

Here we will define some of the common quantum gates used. This is not an exhaustive list. The ones we will define are the Hadamard gate, the $\pi/8$ (rotation) gate, and the controlled-not gate.

\begin{definition}
	The Hadamard gate $H$ acts on a single qubit and corresponds to the following unitary matrix
	\[ H = \frac{1}{\sqrt{2}}\begin{pmatrix}
		1 & 1 \\ 1 & -1
	\end{pmatrix}. \]
\end{definition}

For instance $ H\ket{0} = \ket{\frac{1}{2}} := \frac{1}{\sqrt{2}}\begin{pmatrix}
	1 \\ 1
\end{pmatrix} $ and $ HH\ket{0} = H\frac{1}{\sqrt{2}}\begin{pmatrix}
	1 \\ 1
\end{pmatrix} = \ket{0}. $

It is important to note that the Hadamard gate is both self-adjoint and its own inverse. That is, $HH=HH^{\dag}=I$.

\begin{definition}
	The controlled-not gate, $CNOT$, acts on two qubits and performs the not (bit flip) operation on the second qubit if the first qubit is $|1\rangle$. This equates to the following unitary matrix
	\[ CNOT = \begin{pmatrix}
		1 & 0 & 0 & 0 \\ 0 & 1 & 0 & 0 \\ 0 & 0 & 0 & 1 \\ 0 & 0 & 1 & 0
	\end{pmatrix}. \]
\end{definition}

For instance $CNOT\ket{0} \ket{a} = \ket{0}\ket{a}$ and $CNOT\ket{1} \ket{a} = \ket{1}\otimes \begin{pmatrix}
	\alpha_0 \\ \alpha_1
\end{pmatrix} $

\begin{definition}
	The $\pi/8$ gate, $R_{\pi/4}$, corresponds to a rotation of the $|1\rangle$ qubit by $\pi/4$. The matrix representing this rotation is
	\[ R_{\pi/4} = \begin{pmatrix}
		1 & 0 \\ 0 & e^{i\frac{\pi}{4}}
	\end{pmatrix},\ R_{\pi/4}\ket{a} = \begin{pmatrix}
		\alpha_0 \\ \alpha_1 e^{i\frac{\pi}{4}}
	\end{pmatrix}. \]
\end{definition}

The gate is called the $\pi/8$ gate for historical reasons, even though the gate is a rotation of $\pi/4$. These three gates are important as they form a universal set of gates for two qubits. This means that any classical two bit circuit can be constructed using only these three gates \citep{nielsen2002quantum}.

When working with larger numbers of qubits, to apply one of the above transforms to just some of the qubits, we use identity matrices with block-diagonal matrices of the transform(s). 

For example if given two qubits, we wish to apply the Hadamard transform to only the second one, we can use the following matrix 
\[ H_1 = H\otimes \begin{pmatrix}
	1 & 0 \\ 0 & 1
\end{pmatrix} \]
This process can be used for any number of qubits.

\subsubsection{Measurement}
Now that we have qubits and quantum gates, we need a way to measure the outcome. We call this quantum measurement. This measurement is done according to the $L^2$ norm; for a system of the form $\sum_{x=0}^{2^n-1}\alpha_x |x\rangle$, we will observe $x$ with probability $|\alpha_x|^2$. Specifically, measurement of an $n$ qubit system will do the following:

\[ \sum_{x=0}^{2^n-1}\alpha_x \ket{x} \to i\ \text{ with probability } |\alpha_{i}|^2  \]

The measurement of mixed states yields random outcomes, this is a fundamental aspect of quantum systems.

Additionally, one can choose instead to measure a subset of the quantum circuit, for example the first $m$ qubits. This is no different from the regular measurement, except that the observer will not gain any information about the $n-m$ qubits remaining.

In several quantum computations one uses extra qubits to aid in the computation which are not measured at the end of the computation.

It is important to note that this is destructive measurement. This means that performing the measurement will destroy the superposition of the system. 

Although we can only ensure an outcome with some probability, we can repeat the computation and reduce the probability of error exponentially. Ultimately, we will only ever be sure of an outcome with some (quite high) probability. In most cases this is good enough.

When performing analysis of a quantum algorithm, it is important to include the number of times the computation must be repeated to reduce the error sufficiently, since this will increase the time complexity (potentially exponentially if it is required to be repeated an exponential number of times).

To give some examples of  full quantum circuits, it will be useful to demonstrate some essential quantum algorithms.

\subsubsection{Submodules}
In this section we will explain the building blocks of quantum algorithms. 

\subsubsection{Quantum Oracle}
The quantum oracle is used when we want to apply a function $f:\{0,1\}^n\to \{0,1\}$ to a superposition of all elements of $\{0,1\}^n$. Since all transforms in quantum computing are reversible (and indeed unitary) there needs to be some way to keep the information so that the transform can be reversed. Classically we could take $x\to f(x)$, however when performing this transform in quantum computing we do the following

\[ U_f \ket{x}\ket{y}=  \ket{x}\ket{y\oplus f(x)}.  \]

Where $y$ is representing an extra qubit used for this reversibility.

\subsubsection{Quantum Fourier Transform}

The quantum Fourier transform is the quantum version of the Fourier transform. Described in \cite{nielsen2002quantum}, the quantum Fourier transform ($\mathcal{QFT}$) is a linear operator which acts on a vector $\ket{j}$ of size $2^n$ as follows,

\begin{equation}
	\mathcal{QFT} \ket{j} = \frac{1}{2^{n/2}}\sum_{k=0}^{2^n-1}e^{2\pi ijk/2^n}\ket{k}.
\end{equation} 

An expanded representation of the quantum Fourier transform, into the binary expression of $j=j_1j_2\ldots j_n$, is 

\begin{equation}
	\ket{j} \to \frac{1}{2^{n/2}} \left( \ket{0}+e^{2\pi i 0.j_n}\ket{1} \right) \left( \ket{0}+e^{2\pi i 0.j_{n-1}j_n}\ket{1} \right) \ldots \left( \ket{0}+e^{2\pi i 0.j_1\ldots j_n}\ket{1} \right).
\end{equation}	
Here $0.j_n$ is the $n$th binary digit of $j$ divided by 2, likewise for $0.j_{n-1}j_n$ and $0.j_1\ldots j_n$.

The unitary matrix form of the $\mathcal{QFT}$ of $n$ qubits is a matrix $\mathcal{QFT}=\frac{1}{\sqrt{2^n}}(a_{ij})$ with entries $a_{ij}=\omega^{(i-1)(j-1)}$, where $\omega$ is the $2^n$th root of unity.

\subsubsection{Inverse Quantum Fourier Transform}
Also somewhat unsurprisingly the inverse quantum Fourier transform ($\mathcal{QFT}^{-1}$) is the inverse of the quantum Fourier transform. That is, $\mathcal{QFT}\ \mathcal{QFT}^{-1}=\mathcal{QFT}^{-1}\ \mathcal{QFT}=I$ the identity.

\begin{equation}
	\mathcal{QFT}^{-1}  \left( \frac{1}{2^{n/2}}\sum_{k=0}^{2^n-1}e^{-2\pi ijk/2^n}\ket{k}\right) =\ket{j}
\end{equation}

\subsubsection{Phase Estimation}
\label{phase}
Given a unitary transform $U$ with an eigenvector $\ket{x}$ and an eigenvalue $e^{2\pi i \omega}$, phase estimation, as described in \cite{nielsen2002quantum}, allows us to estimate the value of $\omega$. The algorithm operates on two registers, first a register of $t$ qubits, where $t$ depends on the number of bits of accuracy required and the probability of success required. Specifically, for accuracy of $m$ and probability of success $1-\epsilon$, we choose $t=m+\lceil \log (2+\frac{1}{2\epsilon})\rceil$. The second register is $\ket{x}$. The algorithm has four steps, first performing the Hadamard operation on the first register; then performing a controlled version of the unitary operation $U$ on the first register; then performing the inverse Fourier transform, as mentioned above, on the first register; then lastly performing measurement on the first register.

Let $0.\omega_1\ldots\omega_t$ be the binary expression of the first $t$ digits of $\omega$, then algebraically the steps above are as follows,
\begin{align*}
	\ket{0}^{\otimes t} \ket{x} &\to \frac{1}{2^{t/2}}\sum_{k=0}^{2^t -1}\ket{k}\ket{x} &\text{Hadamard} \\
	&\to \frac{1}{2^{t/2}}\sum_{k=0}^{2^t -1}e^{2\pi i \omega k }\ket{k}\ket{x} &\text{Controlled-}U \\
	&= \frac{1}{2^{t/2}} \left( \ket{0}+e^{2\pi i 2^{t-1}\omega}\ket{1} \right) \left( \ket{0}+e^{2\pi i 2^{t-2}\omega}\ket{1} \right) \\
	&\ \ldots \left( \ket{0}+e^{2\pi i 2^{0}\omega}\ket{1} \right) \ket{x} \\
	&= \frac{1}{2^{t/2}} \left( \ket{0}+e^{2\pi i 0.\omega_t}\ket{1} \right) \left( \ket{0}+e^{2\pi i 0.\omega_{t-1}\omega_t}\ket{1} \right) \\
	&\ \ldots \left( \ket{0}+e^{2\pi i 0.\omega_1\ldots\omega_t}\ket{1} \right) \ket{x} \\
	&\to \ket{\tilde{\omega}}\ket{x} &\text{Inverse Fourier transform} \\
	&\to \tilde{\omega} &\text{Measurement on first register} 
\end{align*}

The Quantum circuit for Phase estimation is as follows,

\begin{figure}[H]
\[  \Qcircuit @C=2em @R=1.1em {
  \lstick{\ket{0}}    & \gate{H} &  \qw &\qw & \lstick{\cdots} & \ctrl{4} &\multigate{3}{\mathcal{QFT}^{-1}} &\meter & \cw \\
  \lstick{\vdots\ \ }   & \gate{H} & \qw &\qw & \lstick{\cdots} &\qw & \ghost{\mathcal{QFT}^{-1}} &\meter & \cw \\
  \lstick{\ket{0}}    & \gate{H} &  \qw   & \ctrl{2} & \lstick{\cdots} &\qw & \ghost{\mathcal{QFT}^{-1}} &\meter & \cw \\
  \lstick{\ket{0}}  & \gate{H} &  \ctrl{1} &\qw & \lstick{\cdots} &\qw & \ghost{\mathcal{QFT}^{-1}} &\meter & \cw \\
  \lstick{\ket{x}}   &  & \gate{U^{2^0}} & \gate{U^{2^1}} & \lstick{\cdots} &\gate{U^{2^{t-1}}} & \qw }
\]
 \caption{Quantum Circuit for the Phase estimation algorithm \citep{nielsen2002quantum,wiki:qca}}
\end{figure}

Phase estimation is an essential part of the Quantum Counting Algorithm and many other quantum algorithms.

\subsection{Quantum Algorithms}
In this section we will describe some quantum algorithms. The general setup of all of these algorithms is as follows:
\begin{itemize}
\item Take some initial state such as $|0\rangle^{\otimes n}$
\item Quantumize to create a uniform superposition over all possible qubits, often done with the Hadamard gate $H^{\otimes n}$
\item Perform computation of some function in simultaneous states of this superposition
\item Uncompute the superposition, often done with the Hadamard gate or the inverse Quantum Fourier transform
\item Measurement of some or all of the circuit
\end{itemize}
We will first describe the setup of each algorithm, then give an intuitive explanation to each and why they are superior to their classical counterpart. This will be followed by a formal algebraic description of the algorithms, then a demonstration of the quantum circuit that would be used. This is by no means a complete list of Quantum algorithms; it is however a list of all algorithms which are relevant to this technical report.

\subsubsection{Deutsch-Jozsa Algorithm}
\label{DJsec}
The Deutsch-Jozsa Algorithm is the first example of an exponential ``quantum-speedup''. Imagine we are given a function $f:\ \{ 0,1\}^n\to \{0,1\}$ that has the property that either all values map to $0$, or half of them do. Our objective is to determine whether every value maps to 0, or half of them do. To check this classically, one must perform at most $2^{n-1}+1$ function evaluations. This is because the moment the function outputs a 1 we know that $f$ outputs 1 on half the inputs. The Deutsch-Jozsa Algorithm requires only $1$ function evaluation. 

An exponential speedup comes from evaluating the function in a superposition of every input. To achieve this, we use the Hadamard gate which takes the initial state to a superposition of every possible state, with equal amplitude. Then we use oracle (function) $f$ once in the oracle transform $U_f$ which is defined as follows:

\[ U_f \ket{x}\ket{y}=  \ket{x}\ket{y\oplus f(x)}  \]
where $\oplus$ is addition modulo 2 (XOR). Then after a simplification a second Hadamard gate is used, and finally after a simplification step, measurement of the first $n$ bits is performed. If $f(x)=0$ for all $x$, then the measurement will be 1 with probability 1, and if $f(x)=0$ on half the inputs and $f(x)=1$ on the other half of inputs, then the measurement will be 0 with probability 1. Below we have an algebraic description.

\begin{align*}
	|0\rangle^{\otimes n}|1\rangle &\to  \frac{1}{\sqrt{2^{n+1}}} \sum_{x=0}^{2^n -1}|x\rangle (|0\rangle - |1\rangle) &\text{Hadamard } H^{\otimes n} \otimes H \\
	&\to \frac{1}{\sqrt{2^{n+1}}} \sum_{x=0}^{2^n -1}|x\rangle (|f(x)\rangle - |1\oplus f(x) \rangle) &f\text{ oracle} \\
	&= \frac{1}{\sqrt{2^{n+1}}} \sum_{x=0}^{2^n -1}(-1)^{f(x)}|x\rangle (|0\rangle - |1\rangle) &\text{since } f(x)=0,1 \\
	&\to \frac{1}{2^n} \sum_{x=0}^{2^n -1}(-1)^{f(x)} \left[ \sum_{y=0}^{2^n -1}(-1)^{x\cdot y} |y\rangle \right] \ket{1} &\text{Hadamard } H^{\otimes n} \otimes H \\
	&= \frac{1}{2^n} \sum_{y=0}^{2^n -1} \left[ \sum_{x=0}^{2^n -1}(-1)^{x\cdot y +f(x)}  \right] |y\rangle \ket{1} &\text{Re-ordering}\\
	&\to \left| \frac{1}{2^n} \sum_{x=0}^{2^n -1}(-1)^{f(x)}  \right|^2 &\text{Measurement on first } n \text{ qubits} \\
	&= \begin{cases}
		1 &\text{if } f(x)=0\ \forall x\in \{ 0,1\}^n \\
		0 &\text{if } f(x) = 0 \text{ for half the }  x\in \{ 0,1\}^n
	\end{cases}
\end{align*}

The quantum circuit below is exactly the transforms described above.

\begin{figure}[H]\[
 \Qcircuit @C=2em @R=1.1em {
  \lstick{\ket{0}} & /^{n} \qw & \gate{H^{\otimes n}} & \multigate{1}{U_f} & \gate{H^{\otimes n}} &\meter & \cw \\
  \lstick{\ket{1}} & \qw     & \gate{H}             & \ghost{U_f}        & \qw }
\]
\caption{Quantum circuit for the Deutsch-Jozsa Algorithm \citep{nielsen2002quantum}}
\end{figure} 
 
It is important to note that one may use the Deutsch-Jozsa Algorithm with an oracle $f:\{ 0,1\}^n\to\{0,1\}$ for which the proportion of inputs which map to $0$ or $1$ is some arbitrary (unknown) number. However the resulting algorithm would only be correct with some probability. In this case, the quantum algorithm will have to be repeated to find the proportion with high accuracy.

 The last line of the algorithm would become
\[ \left| \frac{1}{2^n} \sum_{x=0}^{2^n -1}(-1)^{f(x)}  \right|^2  = \left( \frac{2^n-2L}{2^n}\right)^2 \] where $L$ is the number of $x$ such that $f(x)=1$.

In some cases the number of times the algorithm will have to be repeated will be exponentially large in $n$, for example if $f(x)=1$ for only polynomial many $x$, in which case a classical algorithm may be just as effective.

With some algebra (which can be done classically)
we can get \[ \frac{1-\sqrt{\left| \frac{2^n-2L}{2^n}\right|^2} }{2}= \frac{L}{2^n} \]
 which is the desired result of the algorithm modified for any $f:\{ 0,1\}^n\to\{0,1\}$.

\begin{lemma} \label{DJlem}
	For any $f:\{0,1\}^l\to\{0,1\}$ with $L$ being the number of $x\in\{0,1\}^l$ such that $f(x)=1$. Let $X_1,\ldots,X_m\in\{0,1\}$ be the outputs from the modified Deutsch-Jozsa algorithm. The number of trials, $m$, of the modified Deutsch-Jozsa algorithm requires to compute $X=\frac{\sum_{i=1}^m X_i}{m}$ such that $P\left( \left| X-\frac{L}{2^n}\right| <\epsilon\right) <2e^{-2k}$
 is $O\left( \frac{k}{\epsilon^2}\right)$.
\end{lemma}
\begin{proof}
	Let $X_1,\ldots,X_m$ be the outputs of $m$ trials of the modified Deutsch-Jozsa algorithm, note that these are iid. Then let $X = \frac{\sum_{i=1}^m X_i}{m}$ be the empirical mean the of the $X_i$'s. 
	By Hoeffding inequality we have
	\[ Pr(|X-\mathbb{E}[X]|\geq \epsilon ) \leq 2e^{-2m^2\epsilon^2} \]
	That is, after $m$ trials the probability of having an error of $\epsilon$ is bounded by $2e^{-2m\epsilon^2}$. Therefore setting $m=\frac{k}{\epsilon^2}$ for some $k$ we can guarantee the probability of having absolute error of at most $\epsilon$ is bounded by $2e^{-2k}$.
\end{proof}
 
\subsubsection{Shor's Algorithm}
Shor's algorithm \citep{shor1994algorithms} is perhaps the most famous Quantum algorithm. While being one of the few exponential speedups over current classical algorithms, its fame comes more from the problem which it solves: prime factorisation.

Given an $N\in\mathbb{N}$ such that $N=pq$ for $p,q$ prime, find $p$ (or $q$). It is only required to find one of the factors since division is easy in the sense that there is a fast classical algorithm. This problem is quite important for current cryptography, which makes the assumption that this problem is hard in the sense that it is slow to solve on classical computers (It is important to note that it is still possible that there exists a fast classical algorithm for prime factorisation, since we have not proved that the fastest algorithm is not in \textbf{P}; doing so would immediately imply $\textbf{P}\neq\textbf{NP}$).

Shor's algorithm is made up of two parts: a classical part which reduces factoring to order finding, and a quantum part which finds the order. In the classical part, the algorithm checks whether $N$ is even, and calculates the GCD (greatest common divisor) of numbers. Both of these have fast classical algorithms.

The classical algorithm for converting factoring to ordering finding is as follows:

\begin{center}
\begin{algorithm}[H]
\SetAlgoLined
Given $N$\;
\uIf{$N$ is even}{return 2\;}{\uIf{$N=a^b$ for $a\geq 1,\ b\geq 2$}{return $a$\;}{Randomly pick some $x$ in $[1,N-1]\subseteq \mathbb{N}$\; \uIf{$\gcd(x,N)>1$}{return $\gcd(x,N)$\;}{Use order finding to find $r$ such that $x^r \equiv 1\mod N$\; \eIf{$r$ is even and $x^{r/2} \equiv -1\mod N$}{\uIf{$\gcd(x^{r/2}+1,N)>1$}{return $x^{r/2}+1$\;}\uElseIf{ $\gcd(x^{r/2}-1,N)>1$}{return $x^{r/2}-1$\;}\Else{Restart} }{Restart} } }
}
 \KwResult{A factor of $N$ with probability $O(1)$ }
 \caption{Classical Order finding \citep{nielsen2002quantum}}
\end{algorithm}
\end{center}

A fast classical algorithm to find the order has not been found. However the second (Quantum) part of Shor's algorithm can find the order fast.

Let $U_{x,N}\ket{a}\ket{b} = \ket{a}\ket{x^ab \mod N}$ for some $x$ (chosen in the classical algorithm part), and let $m=2\lceil \log(N)\rceil +1+ \lceil \log \left( 2+\frac{1}{2\epsilon} \right) \rceil$ for some error term $\epsilon$ (such as 1/10). 

We need a to use the inverse $\mathcal{QFT}$ for Shor's algorithm.

The quantum algorithm for finding the order is as follows,
\begin{figure}[H]
\begin{align*}
	\ket{0}^{\otimes m}\ket{1}^{\otimes n}&\to  \frac{1}{\sqrt{2^m}}\sum_{j=0}^{2^m-1} \ket{j}\ket{1}^{\otimes n} &\text{Hadamard } H^{\otimes m}\text{ to first }m \text{ bits} \\
	&\to \frac{1}{\sqrt{2^m}}\sum_{j=0}^{2^m-1} \ket{j}\ket{x^j \mod N} &U_{x,N} \text{ applied} \\
	&= \frac{1}{\sqrt{r2^m}}\sum_{s=0}^{r-1}\sum_{j=0}^{2^m-1}e^{2\pi i s j/r} \ket{j}\ket{u_s} &\text{ Fourier decomposition} \\
	&\to \frac{1}{\sqrt{r}}\sum_{s=0}^{r-1} \ket{\widetilde{s}/r}\ket{u_s} &\text{Apply inverse Fourier transform to first $m$ qubits} \\
	&\to \widetilde{s}/r &\text{Measurement on first $m$ qubits} \\
	&\to r &\text{Continued fraction algorithm*}
\end{align*}
\caption{Order finding quantum part of Shor's algorithm \citep{nielsen2002quantum}}
\end{figure}

Using the inverse Fourier transform $\widetilde{s}/r$ is an approximation of the phase of $e^{2\pi i s j/r}$. The continued fraction algorithm is used to find an $r$ such that the output $\widetilde{s}/r$ is irreducible.

 Shor's algorithm takes $O((\log(N))^3)$ time, as opposed to the fastest known fully classical algorithm for prime factorisation which takes $O(N^{1/4})$ time.

The circuit for Shor's algorithm is as follows: since the order $r$ (actually $\widetilde{l}/r$)  we are trying to find is greater than 1, we will require multiple registers to measure it.
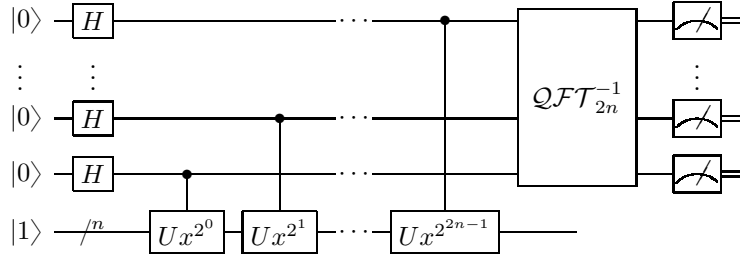
\begin{figure}[H]\[
 \Qcircuit @C=.7em @R=.7em {
  \lstick{\ket{0}}    & \gate{H} & \qw & \qw               & \qw               & \qw & \cdots & & \ctrl{4}               & \multigate{3}{\mathcal{QFT}_{2n}^{-1}} & \qw  & \meter & \cw \\
  \lstick{\vdots\ \ } & \vdots   &     &                   &                   &     &        & &                        &    \pureghost{\mathcal{QFT}_{2n}^{-1}} &      & \vdots &     \\
  \lstick{\ket{0}}    & \gate{H} & \qw & \qw               & \ctrl{2}          & \qw & \cdots & & \qw                    &        \ghost{\mathcal{QFT}_{2n}^{-1}} & \qw  & \meter & \cw \\
  \lstick{\ket{0}}    & \gate{H} & \qw & \ctrl{1}          & \qw               & \qw & \cdots & & \qw                    &        \ghost{\mathcal{QFT}_{2n}^{-1}} & \qw  & \meter & \cw \\
  \lstick{\ket{1}}    & /^n \qw  & \qw & \gate{U{x^{2^0}}} & \gate{U{x^{2^1}}} & \qw & \cdots & & \gate{U{x^{2^{2n-1}}}} & \qw
 }
\]
\caption{Quantum circuit for Shor's algorithm \citep{wiki:shor}}
\end{figure}

There have been some improvements on Shor's algorithm, e.g. by \cite{bernstein2017post} who present a new algorithm that is faster than Shor's algorithm in most cases.

\subsubsection{Grover Search}
 Grover's search \citep{grover1996fast}, formally described in \cite{nielsen2002quantum}, takes a function $f$ such that there is at least one $s$ such that $f(s)=1$, a set $S=\{ 0,1\}^n$ of inputs of size $|S|=N=2^n$, and is able to find an $s\in S$ which satisfies $f(s)=1$ in $O(\sqrt{N})$ time. Classically this kind of search has to take at least $N$ steps, and therefore take $O(N)$ time. This speedup, though not exponential, is quite significant considering the generality of this algorithm. It was also proven that this is the maximal speedup possible for this problem \citep{zalka1999grover}.

Much like our $f$ oracle mentioned in the Deutsch-Jozsa algorithm, let $U$ be an oracle defined as $U\ket{x} \ket{y} = |x\rangle |y\oplus f(x)\rangle$ where $f(x)=1$ if $x$ is a solution to the search problem and 0 otherwise. The algorithm also uses conditional phase shift which takes $\ket{0}$ to $\ket{0}$ and $\ket{x}$ to $-\ket{x}$ for $x>0$. 
The Grover iteration used in Grover's search is defined as a product of 4 gates. First the $U$ gate, followed by a Hadamard gate, then conditional phase shift, and then another Hadamard gate. Together the Grover iteration looks like
\begin{align*}
	G &= (H^{\otimes n} (2\ket{0}^{\otimes n}\langle 0|^{\otimes n} - I_n) H^{\otimes n} )U_{}
\end{align*}
\cite{nielsen2002quantum}.

To demonstrate how the Grover operator is able to give the desired answer, a geometric analysis is quite useful. Let $M$ denote the number of solutions to $f(s)=1$, that is $M = | \{s\in S:f(s) = 1 \}|$.

Let $\ket{\beta} = \frac{1}{\sqrt{M}}\sum_{x:f(x)=1}\ket{x}$ be the vector of all $M$ solutions, and $\ket{\alpha} = \frac{1}{\sqrt{N-M}}\sum_{x:f(x)\neq1}\ket{x}$, then we can write the uniform state as \[ \frac{1}{\sqrt{2^n}}\sum_{x=0}^{2^n-1} \ket{x} = \sqrt{\frac{N-M}{N}} \ket{\alpha} + \sqrt{\frac{M}{N}}\ket{\beta}. \]
 The oracle transform reflects $\ket{\beta}$ about $\ket{\alpha}$; mathematically we can write this as  
 \[ U_{\omega} (a\ket{\alpha}+b\ket{\beta})\left( \frac{\ket{0}-\ket{1}}{\sqrt{2}} \right)  =  a\ket{\alpha}\left( \frac{\ket{0}-\ket{1}}{\sqrt{2}} \right) +b\ket{\beta}\left( \frac{\ket{1}-\ket{0}}{\sqrt{2}} \right)  = (a\ket{\alpha} - b\ket{\beta} )\left( \frac{\ket{0}-\ket{1}}{\sqrt{2}} \right) \]

The transform $(H^{\otimes n} (2\ket{0}\langle 0| - I) H^{\otimes n} )$ is a reflection about $\frac{1}{\sqrt{2^n}}\sum_{x=0}^{2^n-1} \ket{x}$.  Performing these two reflections together gives a rotation.

Let $\cos \frac{\theta}{2} = \sqrt{\frac{N-M}{N}}$, then we have that $\sin\frac{\theta}{2} = \sqrt{\frac{M}{N}} $ and we can re-write the uniform state as,
\[ \frac{1}{\sqrt{2^n}}\sum_{x=0}^{2^n-1} \ket{x} = \cos \frac{\theta}{2} \ket{\alpha} + \sin\frac{\theta}{2}\ket{\beta}. \]
Then applying the Grover iteration to both sides we get
\begin{align*} G\left(\frac{1}{\sqrt{2^n}}\sum_{x=0}^{2^n-1} \ket{x}\right) &= (H^{\otimes n} (2\ket{0}^{\otimes n}\langle 0|^{\otimes n} - I_n) H^{\otimes n} )U_{\omega} \left(\cos \frac{\theta}{2} \ket{\alpha} + \sin\frac{\theta}{2}\ket{\beta} \right) \\
&= (H^{\otimes n} (2\ket{0}^{\otimes n}\langle 0|^{\otimes n} - I_n) H^{\otimes n} ) \left(\cos \frac{\theta}{2} \ket{\alpha} - \sin\frac{\theta}{2}\ket{\beta} \right) \\ 
&= \cos \left( \frac{3\theta}{2}\right) \ket{\alpha} + \sin \left(\frac{3\theta}{2}\right)\ket{\beta} \end{align*}

and applying the iteration $k$ times leads to 

\[ G^k\left(\frac{1}{\sqrt{2^n}}\sum_{x=0}^{2^n-1} \ket{x}\right) = \cos \left( \frac{2k\theta +\theta}{2}\right) \ket{\alpha} + \sin\left( \frac{2k\theta +\theta}{2}\right)\ket{\beta}. \]

Thus we perform the iteration a number of times so that $\sin\left( \frac{2k\theta +\theta}{2}\right)$ is close to 1, which leads to 

\[ k = \left\lceil \frac{\pi}{4}\sqrt{\frac{N}{M}} \right\rceil. \]

This can be derived by
\begin{align*}
	\sin\left( \frac{2k\theta +\theta}{2}\right) &\approx 1 \\
	\frac{2k\theta +\theta}{2} &\approx \frac{\pi}{2} \\
	\theta\frac{2k+1}{2} &\approx \frac{\pi}{2} \\
	2k+1 &\approx \frac{\pi}{\theta} \\
	k &\approx \frac{\pi}{2\theta} - \frac{1}{2} \\
	k &\approx \frac{\pi}{4} \sqrt{\frac{N}{M}} - \frac{1}{2}
\end{align*}

We could include the $+2m\pi$ in the second line, however we are trying to use the least number of iteration steps so we pick $m=0$ in this case. Now that we have shown how the Grover iteration is able to find the desired answer, the algorithm for producing the single unique $x'$ such that $f(x')=1$ (thereby the case when $M=1$) can be defined as follows:
\begin{figure}[H]
\begin{align*}
	\ket{0}^{\otimes n}\ket{0}&\to  \frac{1}{\sqrt{2^n}}\sum_{x=0}^{2^n-1} \ket{x}(\ket{0} - \ket{1}) &\text{Hadamard} \\
	&\text{then repeat the Grover iteration } \lceil (\pi \sqrt{N} /4) \rceil \text{ times} \\
	&\to ((H^{\otimes n} (2\ket{0}\langle 0| - I) H^{\otimes n} )U_ )^{ \lceil (\pi \sqrt{N} /4) \rceil} \left( \frac{1}{\sqrt{2^n}}\sum_{x=0}^{2^n-1} \ket{x}(\ket{0} - \ket{1})\right) \\
	&= G^{ \lceil (\pi \sqrt{N} /4) \rceil} \left( \frac{1}{\sqrt{2^n}}\sum_{x=0}^{2^n-1} \ket{x}(\ket{0} - \ket{1})\right) \\
	&\approx \ket{\beta} \left( \frac{\ket{0}-\ket{1}}{2}\right) \\
	&= \ket{x'} \left( \frac{\ket{0}-\ket{1}}{2}\right) \\
	&\to x' &\text{Measurement on first $n$ qubits}
\end{align*}
\caption{Quantum Search Algorithm  \citep{nielsen2002quantum}}
\end{figure}

To produce a quantum circuit, we can just write out each transform used in order.

\begin{figure}[H]
	\[ \Qcircuit @C=1em @R=.7em {
 & & & & & \ustick{\text{Grover operator}} & & & & & & \\
\lstick{\ket{0}} & /^n \qw & \gate{H^{\otimes n}} & \multigate{1}{U} & \gate{H^{\otimes n}} & \gate{2 \ket{0}\langle 0| - I}         & \gate{H^{\otimes n}} & \qw & \cdots & & \meter & \cw \\
  \lstick{\ket{1}} & \qw     & \gate{H}             & \ghost{U_\omega}        & \qw & \qw & \qw & \qw & \cdots & \\
 & & & & & \text{Repeat $O(\sqrt{N})$ times} & & & & & &  } \]
 \caption{Quantum Circuit for Grover's algorithm \citep{nielsen2002quantum,wiki:grov}}
\end{figure}

It may not be immediately obvious what our value of $M$ should be, that is, how many solutions there may be. It turns out that we can solve this problem with our next algorithm.

\subsubsection{Quantum Counting Algorithm}
The Quantum Counting Algorithm, proposed in \cite{brassard1998quantum} and described in \cite{nielsen2002quantum}, is a combination of Grover search and phase estimation. Given an oracle indicator function $f_B:A\to \{0,1\}$ of $B\subseteq A$, with $|A|=N =2^n$, the Quantum Counting Algorithm finds $M=|B|$.

To do so, the Quantum Counting Algorithm finds a solution $\theta$ to the equation \begin{equation}
	\sin^2 \left( \frac{\theta}{2} \right) =  \frac{M}{2N}
\end{equation} then solves for $M$.

Phase estimation, described in Section \ref{phase} and \cite{nielsen2002quantum}, is a subroutine used in quantum algorithms to estimate the phase of the eigenvalue of some unitary operator (in this case $G$) to some precision. Phase estimation relies on the fact that when the eigenvalue is written in the form $e^{2\pi i j\phi}$ for phase $\phi$, the inverse Fourier transform will transform $\frac{1}{\sqrt{N}}\sum_{j=0}^{N-1}e^{2\pi i j\phi}\ket{j}$ to an approximation of $\phi$ in the form $\ket{\tilde{\phi}}$, where $\tilde{\phi}$ is the binary approximation of $\phi$.

To achieve $m$ bits of accuracy of $\theta$ with probability $1-\epsilon$, the algorithm works on two registers. The first register is of size $t=m+\lceil \log (2+\frac{1}{2\epsilon})\rceil$, and the second register of size $n+1$.

The algorithm is much like the phase estimation described in Section \ref{phase}.

\begin{align*}
	\ket{0}^{\otimes t} \ket{0}^{n+1} &\to \frac{1}{2^{t/2}}\sum_{k=0}^{2^t -1}\ket{k}\frac{1}{2^{(n+1)/2}}\sum_{s=0}^{2^{n+1} -1}\ket{s} &\text{Hadamards} \\
	&\to \frac{1}{2^{t/2}}\sum_{k=0}^{2^t -1}e^{2\pi i \phi k }\ket{k}\frac{1}{2^{(n+1)/2}}\sum_{s=0}^{2^{n+1} -1}\ket{s} &\text{Controlled-}G \\
	&= \frac{1}{2^{t/2}} \left( \ket{0}+e^{2\pi i 2^{t-1}\phi}\ket{1} \right) \left( \ket{0}+e^{2\pi i 2^{t-2}\phi}\ket{1} \right) \\
	&\ \ldots \left( \ket{0}+e^{2\pi i 2^{0}\phi}\ket{1} \right)\frac{1}{2^{(n+1)/2}}\sum_{s=0}^{2^{n+1} -1}\ket{s}  \\
	&= \frac{1}{2^{t/2}} \left( \ket{0}+e^{2\pi i 0.\phi_t}\ket{1} \right) \left( \ket{0}+e^{2\pi i 0.\phi_{t-1}\phi_t}\ket{1} \right) \\
	&\ \ldots \left( \ket{0}+e^{2\pi i 0.\phi_1\ldots\phi_t}\ket{1} \right) \frac{1}{2^{(n+1)/2}}\sum_{s=0}^{2^{n+1} -1}\ket{s}  \\
	&\to \ket{\tilde{\phi}}\frac{1}{2^{(n+1)/2}}\sum_{s=0}^{2^{n+1} -1}\ket{s}  &\text{Inverse Fourier transform} \\
	&\to \tilde{\phi} &\text{Measurement on first register} 
\end{align*}

The circuit of the algorithm is as follows,

\begin{figure}[H]
\[  \Qcircuit @C=2em @R=1.1em {
  \lstick{\ket{0}}    & \gate{H} &  \qw &\qw & \lstick{\cdots} & \ctrl{4} &\multigate{3}{\mathcal{QFT}_m^{-1}} &\meter & \cw \\
  \lstick{\vdots\ \ }   & \gate{H} & \qw &\qw & \lstick{\cdots} &\qw & \ghost{\mathcal{QFT}_m^{-1}} &\meter & \cw \\
  \lstick{\ket{0}}    & \gate{H} &  \qw   & \ctrl{2} & \lstick{\cdots} &\qw & \ghost{\mathcal{QFT}_m^{-1}} &\meter & \cw \\
  \lstick{\ket{0}}  & \gate{H} &  \ctrl{1} &\qw & \lstick{\cdots} &\qw & \ghost{\mathcal{QFT}_m^{-1}} &\meter & \cw \\
  \lstick{\ket{0}^{\otimes n+1}}   & \gate{H^{\otimes n+1}} & \gate{G^{2^0}} & \gate{G^{2^1}} & \lstick{\cdots} &\gate{G^{2^{m-1}}} & \qw }
\]
 \caption{Quantum Circuit for the Quantum Counting algorithm \citep{nielsen2002quantum,wiki:qca}}
\end{figure}

If we choose $m=\lceil n/2\rceil +1$ and $\epsilon$ sufficiently small (such as $1/10$), then the algorithm will take $O(\sqrt{N})$ Grover iterations. This means that the function $f$ will only be called $O(\sqrt{N})$ times. Note that this is in contrast to a classical (deterministic or probabilistic) algorithm which will take $O(N)$ oracle calls to achieve the same accuracy. Formally,

\begin{theorem}[Quantum Counting Correctness]	
\label{quancountcorr}
	Given a function $f:\{0,1\}^n\to \{0,1\}$ such that $M=|\{x\in \{0,1\}^n:f(x)=1\}|$ and $\sin^2\left( \frac{\theta}{2}\right)=\frac{M}{2N}$, to find $\theta$ with $m$ bits of accuracy, with probability $1-\epsilon$ the Quantum Counting Algorithm requires $O(m+n+\lceil \log(2+\frac{1}{2\epsilon})\rceil )$ registers and $O(\sqrt{N})$ time.
\end{theorem}
\begin{proof}
The majority of this proof is from \cite{nielsen2002quantum}. Given $m$ and $\epsilon$ set up the first register with $\lceil \log(2+\frac{1}{2\epsilon})\rceil$ qubits, and the second register with $n+1$ qubits. Use the Hadamard gate on the second register to take it to the superposition of $\frac{1}{\sqrt{2^{n+1}}}\sum_{x=0}^{2^{n+1}-1}\ket{x}$. Like the Grover search algorithm let $\ket{a}$ and $\ket{b}$ represent the eigenvectors of the Grover iteration with eigenvalues $e^{i\theta}$ and $e^{i(2\pi - \theta)}$ respectively. The superposition of the second register can be written in the form of $\ket{a}$ and $\ket{b}$. The phase estimation algorithm \ref{phase} allows us to estimate the phase of the eigenvalues of $\ket{a}$ or $\ket{b}$, that is $\theta$ or $2\pi-\theta$, to within $|\triangle\theta|\leq 2^{-m}$ with probability at least $1-\epsilon$. Therefore we are able to determine $\theta$ to an accuracy of $2^{-m}$ with probability at least $1-\epsilon$.
\end{proof}

\subsubsection{Harrow-Lloyd Algorithm for Linear equations}
Given some $N\times N$ matrix $A$ and some vector $b$, finding the solution $x$ to the equation $Ax=b$ is known as the linear equation problem. Classically this can be done in many ways, such as matrix inversion (finding $A^{-1}$ such that $x=A^{-1}b$). Classically the fastest algorithm takes $O(N\kappa)$ time, where $\kappa$ is the condition number of the matrix $A$. The Harrow-Lloyd algorithm \citep{harrow2009quantum} is able to achieve an exponential speedup in $N$ by taking $O(\log(N)\kappa^2 )$ time, if $\kappa = O(1)$. Note that when $\kappa = O(N)$ this algorithm provides no speedup.

At this point the reader may question the existence of the algorithm since to output an $N$ long vector $x$, one must use at least $N$ steps. This is correct, however, if one is interested in some property of $x$, such as  $||Mx||_{tr}$ for some matrix $M$, it will provide an exponential speedup over classical methods. The procedure relies on the quantum phase estimation and hamiltonian simulation, for both of which there are fast quantum algorithms.

\subsection{Quantum Complexity Theory}

Quantum complexity theory is developed in the landmark paper by \cite{bernstein1997quantum}, which expands on Deutsch's work by defining what it means for a QTM to be well-formed. Then showed that well-formedness of a QTM is equivalent to having a unitary time evolution operator (as Quantum Physics requires). Then, it goes on to demonstrate how to (theoretically) construct an efficient universal QTM. This is done by proving some results about reversible Turing Machines, which in turn apply to Quantum Turing Machines, since unitary transforms are by definition reversible. Then \cite{bernstein1997quantum} used those results to construct a looping and branching process required by an efficient QTM. To show how to decompose a unitary transform, \cite{bernstein1997quantum} defined a class of matrices called near-trivial.  \begin{definition}
	A near-trivial matrix is one that is the identity with either a single diagonal phase shift, $e^{i\theta}$, or a rotational block of the form \[\begin{pmatrix}
	\cos (\theta ) & -\sin (\theta ) \\
	\sin (\theta ) & \cos (\theta )
\end{pmatrix} \]  for some $\theta \in [0,2\pi)$.
\end{definition}  It was then proven that there exists a (deterministic) $poly(log(\frac{1}{\epsilon}))$ algorithm to decompose any unitary matrix into near-trivial matrices. \cite{bernstein1997quantum} used $\theta = \mathcal{R}:= 2\pi \sum_{k=1}^{\infty} 2^{-2^k}$ and showed that one can efficiently simulate any QTM with near-trivial transforms with this choice of $\mathcal{R}$. Quantum Turing Machines of the above form, with $\theta = \mathcal{R}$, are hereby referred to as $QTM_{BV}$.

All of this led to the definition of the Quantum complexity classes of $\mathbf{BQP}$ (Bounded Error Quantum Polynomial time), an analogue of $\mathbf{BPP}$, and $\mathbf{EQP}$ (Exact Quantum Polynomial Time). \begin{definition}
	$\mathbf{BQP}$ is defined as the set of languages that are accepted with probability $\frac{2}{3}$ by some polynomial time Quantum Turing Machine.
\end{definition} The class $\mathbf{EQP}$ is defined in a similar way, however the language has a probability of being accepted of 1. To compare Quantum complexity classes to classical complexity classes, \cite{bernstein1997quantum} proved $\mathbf{P}\subseteq \mathbf{EQP}$, that $\mathbf{BPP} \subset \mathbf{BQP}$ and that $\mathbf{BQP} \subseteq \mathbf{PSPACE}$. They then went on to prove that there exist problems which are in $\mathbf{BQP}$ but are not in $\mathbf{BPP}$, showing that Quantum Computing has strict advantages over classical deterministic (or probabilistic) computing.

\subsection{Quantum Computability}

Quantum computability was expanded on by \cite{adleman1997quantum}. Here, the authors used transcendental number theory and were able to produce some results about the computational power of different classes of Quantum Turing Machines. Similarly to \cite{bernstein1997quantum}, \cite{adleman1997quantum} defined a type of QTM much like the near-trivial QTM, i.e. one whose transforms are near-trivial matricies. Let $\theta \in [0,2\pi)$, then let $QTM_{\theta}$ denote a subset of all $QTMs$, whose matrix (time evolution operator) is block diagonal with each block containing 1, $-1$, or a $2\times 2$ of the form \[\begin{pmatrix}
	\cos (\theta ) & -\sin (\theta ) \\
	\sin (\theta ) & \cos (\theta )
\end{pmatrix}. \] Let $QTM_{\mathbb{Q}}=\bigcup_{q\in\mathbb{Q}}QTM_q$ and let $QTM_{BV}$ be the Quantum Turing Machine defined by \cite{bernstein1997quantum}. \cite{adleman1997quantum} showed that $QTM_{\mathbb{Q}} \equiv QTM_{BV}$ in the sense that they can simulate each other within some small $\epsilon >0$ error, with only a polynomial slowdown. Let $BQP_{\theta}$ denote $\mathbf{BQP}$ as defined above for $QTM_{\theta}$. \cite{adleman1997quantum} proved that this leads to $\mathbf{BQP} = BQP_{\mathbb{Q}}$. However, they were also able to prove that $BQP_{\mathbb{Q}} \subsetneq BQP_{\mathbb{C}}$, and that $BQP_{\mathbb{C}}$ contains sets of arbitrary Turing degrees. Recall that Turing degrees \citep{post1944recursively} are the degree of difficulty of a problem, that is, how unsolvable a problem is. This then leads to the fact that $QTM_{\mathbb{Q}} $ cannot simulate $QTM_{\mathbb{C}}$  with $\epsilon$ error in polynomial time. Additionally, \cite{adleman1997quantum} were able to prove that $\mathbf{BQP}$ and $\mathbf{EQP}$ are all contained within $\mathbf{PP}$. Recall that $\mathbf{PP}$ (Probabilistic polynomial time) is the class of all problems solvable with a probabilistic Turing Machine which is correct with probability more than $\frac{1}{2}$  \citep{gill1977computational}. If we define $EQP_{\theta}$ similar to $BQP_{\theta}$, then for $\theta$ such that $\cos(\theta)$ is poly-computable transcendental, $EQP_{\theta} = P$.

\subsection{Quantum Algorithmic Information Theory}
Little work has been done on Quantum Algorithmic Information Theory compared to other areas of Quantum Computing. Quantum Kolmogorov complexity was first discussed by \cite{berthiaume2000quantum}. \cite{muller2008strongly} expanded upon the idea, and defined it as follows,
\begin{definition}[\cite{muller2008strongly}]
	Given a QTM $M$ and a finite error $\delta >0$, the finite-error Quantum Kolmogorov complexity of a qubit string $\ket{x}$ is
	\[ K^Q_{M,\delta} (x) = \min_p \{ \ell (p) \ :\ ||x - M(p)||_{tr} < \delta \} \] and the approximate-scheme Quantum Kolmogorov complexity of a qubit string $x$ is \[ K^Q_{M} (x) = \min_p \left\{ \ell (p) \ :\ ||x - M(p,k)||_{tr} < \frac{1}{k} \ \forall \ k\in \mathbb{N} \right\} \]
	where $||\cdot ||_{tr}$ is the trace norm, i.e. $||a-b||_{tr} := \frac{1}{2}||a-b||_1$.
\end{definition} 

  	To construct a Quantum version of the invariance theorem, a result from classical Kolmogorov complexity, \cite{muller2008strongly} showed that there exists a Universal Quantum Turing Machine, $\mathfrak{U}$ which satisfies the property that for any other QTM $M$ we have 
\[ K^Q_{\mathfrak{U}} (x) \leq K^Q_M (x) + c_M \] for all qubit strings $x$, where $c_M$ is a constant depending only on $M$, but not on $x$. Then  \cite{muller2008strongly} goes on to prove that for all $\delta,\gamma \in \mathbb{Q}^+$ with $\delta < \gamma$, and for all QTM $M$ we have that \[ K^Q_{\mathfrak{U},\gamma} (x) \leq K^Q_{M,\delta} (x) + c_{M,\gamma,\delta} \] for all qubit strings $x$, where $ c_{M,\gamma,\delta}$ is a constant depending on $M,\gamma$ and $\delta$. This shows that the QTM $\mathfrak{U}$ is ``strongly universal''.

\section{Hardness of Counting}
In this section we will explain the classical hardness of counting, some quantum computing approaches, and the implications of a fast classical or quantum counting algorithm.

\subsection{The class $\#\mathbf{P}$}
Introduced by \cite{valiant1979complexity} while describing the complexity class of computing the permanent of a matrix, the complexity class $\#\mathbf{P}$ (pronounced num P) is the class of counting the number of solutions to an $\mathbf{NP}$ problem. For example the number of Hamiltonian cycles on a graph, or the number of inputs satisfying a Boolean circuit. Formally,

\begin{definition}[\cite{valiant1979complexity}]
	 $\#\mathbf{P}$ is the class of all problems which can be computed by counting Turing Machines of polynomial time complexity.
\end{definition}

Some examples of problems in $\#\mathbf{P}$ are: the number of distinct optimal tours in a travelling salesman problem, the number of unique subsets of a set of integers which sum to zero, and the number of subgraphs which are isomorphic to a given graph.

\subsection{Counting}
The problem of exact counting is defined as follows: Given a function $f:\{0,1\}^n\to \{0,1\}$, we are concerned with finding \[ C(n,f)= \frac{1}{2^n} \sum_{x\in \{0,1\}^n: f(x)=1} 1. \] 
The problem of approximate counting is to find a $k$ such that\[  (1-\epsilon )\cdot C(n,f) \leq k < (1+\epsilon )\cdot  C(n,f) \] for some fixed $\epsilon >0$.

It has been shown that approximate counting can be done probabilistically with runtime polynomial in $n$ and $\frac{1}{\epsilon}$, with an  $\mathbf{NP}$-complete oracle \citep{valiant1979complexity,stockmeyer1983complexity,stockmeyer1985approximation}. 

If there was a Quantum algorithm which could perform approximate counting in polynomial time, then that would imply that $\mathbf{NP}\subseteq \mathbf{BQP}$ since as we just mentioned, approximate counting can be done probabilistically with an  $\mathbf{NP}$-complete oracle. This would immediately give (theoretical) quantum supremacy, which is deemed implausible.

\subsection{Toda's Theorem}
An exceptional result by \cite{toda1991pp} demonstrates the power of the complexity classes $\mathbf{PP}$ and $\#\mathbf{P}$.

\begin{theorem}[Toda's Theorem]
	The entire polynomial hierarchy $\mathbf{PH}$ is contained in $\mathbf{P}^{\mathbf{PP}}$.
\end{theorem}
Using Toda's theorem, we get an immediate corollary which comes from $\mathbf{P}^{\mathbf{PP}}= \mathbf{P}^{\#\mathbf{P}}$.
\begin{corollary}
	The entire polynomial hierarchy $\mathbf{PH}$ is contained in $\mathbf{P}^{\#\mathbf{P}}$.
\end{corollary}

If we had an efficient classical algorithm which could perform exact counting, then from Toda's theorem the polynomial hierarchy would collapse. Since we suspect this is not the case, it is unlikely that there exists a classical algorithm which can compute exact counting.

There have been several attempts to use Quantum computing to solve counting problems more quickly than classical computers, such as the quantum counting algorithm in the previous section. Here, we will give two other attempts which provide significant speedup over classical methods: Boson sampling and Postselection.

\subsection{Boson sampling}
Before getting to Boson sampling, it is important to discuss the permanent.

\subsubsection{The Permanent}
Though less well-known than its cousin the determinant, the permanent \citep{weisstein2006permanent} is an operation on a matrix which involves permutations of the rows and columns of the matrix. Formally,
\begin{definition}
	The permanent of an $n\times n$ matrix $A=(a_{i,j})$  is
	\begin{equation} \label{perm3}
		perm(A) = \sum_{\sigma \in S_n} \prod_{i=1}^n a_{i,\sigma(i)}
	\end{equation}
	where $S_n$ is the set of permutations of $\{1,\ldots,n\}$, also known as the symmetric group.
\end{definition}
Alternate representations:

\begin{equation} \label{perm1}
		perm(A) = (-1)^n \sum_{s\subseteq \{1,\ldots,n\}} (-1)^{|s|} \prod_{i=1}^n \sum_{j\in s} a_{i,j}
	\end{equation}
and
\begin{equation} \label{perm2}
		perm(A) = \mathbb{E}\left[\det \begin{pmatrix}
			\pm \sqrt{a_{1,1}} & \cdots & \pm \sqrt{a_{1,n}} \\
			\vdots & \ddots & \vdots \\
			\pm \sqrt{a_{n,1}} & \cdots & \pm \sqrt{a_{n,n}} \\
		\end{pmatrix}^2 \right].
	\end{equation}
	Here the expectation is over the $2^{n^2}$ $\pm$ combinations.

Using either definition it is quite hard to calculate the permanent, at least naively. In Equation \ref{perm1} the outer sum is summing over the set $\{ s: s\subseteq \{1,\ldots,n\} \}$, and this set is of size $2^n$. As for Equation \ref{perm2}, the expected value is over $2^{n^2}$ different matrices.

It should be unsurprising then that calculating the permanent of a matrix $A$ is a $\#\mathbf{P}$-hard problem and a $\#\mathbf{P}$-complete problem if $A$ is a (0,1)-matrix \citep{valiant1979complexity}. In fact, even approximating the permanent of a matrix $A$ is $\#\mathbf{P}$-hard.

\subsubsection{Boson Sampling}
Boson sampling, defined by \cite{aaronson2011computational}, is a proposed (non-universal) model of quantum computation which involves sampling from a probability distribution of noninteracting bosons. It was shown in \cite{abrams1997simulation} that sampling from bosonic or fermionic distributions can be done efficiently (in polynomial time) using a universal quantum computer. In the process of boson sampling, it is required to calculate the permanent of a given matrix. This comes from the fact that the state probabilities in bosonic systems are permanents of matrices.

In defining boson sampling, \cite{aaronson2011computational} proved that, under some reasonable conjectures, the existence of a classical algorithm which could efficiently compute exact (or approximate) boson sampling would imply that the polynomial hierarchy collapses to the third level, meaning having more than two $\mathbf{NP}$-complete oracles is equivalent to having two $\mathbf{NP}$-complete oracles. It is unlikely (though not impossible) that this is the case.

This comes from the fact that exact (or approximate) boson sampling is a $\#\mathbf{P}$-hard problem, since it involves calculating a permanent, and Toda's theorem.

Thus if large, efficient boson Sampling computers were physically realised, we may be able to solve hard problems such as the exact or approximate counting or computing the permanent, which (likely) have no efficient classical algorithm.

\subsection{Postselection}
 Postselection, first defined in \cite{aaronson2005quantum}, is the process of ignoring all outcomes of a computation in which an event did not occur; selecting specific outcomes after (post) the computation.
 
 The complexity class $\mathbf{PostBQP}$, standing for post-selected bounded-error quantum polynomial time, is the class $\mathbf{BQP}$ with postselection. Unsurprisingly we have that $\mathbf{PostBQP}\supseteq \mathbf{BQP}$.
 
 The power of postselection which Quantum computing comes from the following theorem
 \begin{theorem}[\cite{aaronson2005quantum}]
 	$\mathbf{PostBQP} = \mathbf{PP}$. Which immediately implies $\mathbf{PP}\supseteq \mathbf{BQP}$
 \end{theorem}
 
 An immediate result of this is that $\mathbf{P}^{\mathbf{PostBQP}}=\mathbf{P}^{\#\mathbf{P}}$ and therefore $\mathbf{PH}\subseteq \mathbf{P}^{\mathbf{PostBQP}}$. This means that if there was an efficient way to postselect with a Quantum computer, we would be able to solve many of the problems which are intractable classically, including exact counting.

\section{Induction}

The problem of induction, specifically inductive reasoning, is one of the most well-studied problems in philosophy \citep{de1972probability,popper1957philosophy,popper1985problem,hume2000enquiry,cohen1989introduction,carnap1962logical}. First discussed in \cite{Hume1738-HUMATO-3}, the problem of induction and inductive reasoning is the process of reasoning about a hypothesis, given some evidence (data). Outside of philosophy, many areas such as machine learning, statistics, and economics, rely heavily on induction and have produced many practical solutions to the problem. This is because induction can be used to find ``truth'' in the world. Inductive ability has also for a long time been the way in which we test a human's ability, in the form of IQ tests, and many similar examinations. 

Many of the aforementioned studies have proposed solutions to the induction problem, some of which have found great practical success such as \cite{nasrabadi2007pattern} in statistical machine learning. One proposed solution to the Induction problem, Solomonoff induction \citep{solomonoff1964formal}, has benefits over all other solutions (as well as some downsides); this will be the focus of this section.

\subsection{Approaches to Induction}

\subsubsection{Laplace/Bayesian Approach}
Laplace's approach to the induction problem is a rule of succession  \citep{solomonoff1964formal}, also called Bayes-Laplace estimator. This rule was famously demonstrated with the problem of predicting whether or not the sun would rise tomorrow. The following explanation is from \cite{Hutter:04uaibook}.

This approach relies on Bayes' theorem.
\begin{theorem}[Bayes' Theorem]
	Let $H$ and $D$ be events, then \[ P(H|D) = \frac{P(D|H)P(H)}{P(D)} \]
\end{theorem}

Where $P(H)$ be the prior plausibility of the hypothesis $H$, $P(D|H)$ is the likelihood of the data $D$ under the hypothesis $H$, and $P(H|D)$ is the posterior plausibility of of hypothesis $H$, and $\{ H_i\}$ form an exclusive and exhaustive set of hypothesis'. We also have $P(D) = 
\sum_{i} P(D|H_i) P(H_i)$. 
Then Bayes' theorem in its most useful form
\begin{theorem}[Bayes' Rule]
	\[ P(H|D) = \frac{P(D|H)P(H)}{\sum_{i} P(D|H_i) P(H_i)}. \]
\end{theorem}

Using Bayes' rule, one can use the prior probability to calculate the posterior plausibility of of hypothesis $H$ given some data $D$, and this in turn can be used for inductive reasoning and prediction.

Getting back to predicting whether or not the sun will rise tomorrow. Let $x$ be some finite binary sequence, $x=x_1\ldots x_n\in\mathbb{B}^n$, and let $n_0$ and $n_1$ be the number of 0's and 1's in the sequence respectively. Let each element finite sequence be generated by some probability $\theta\in [0,1]$, that is, $P(x_i=1)=\theta$ for all $i$. Our hypothesis class is then $H_{\theta} = \text{Bernoulli}(\theta)$. The likelihood of an event $y$ be given a hypothesis is $P(y|H_{\theta}) = \theta^{n_1}(1-\theta)^{n_0}$, and using a uniform prior plausibility $P(\theta) = P(H_{\theta} ) = 1$. Lastly the probability of some evidence $P(y) = \frac{n_1!n_0!}{(n+1)!}$. All together with Bayes' rule we get
\[ P(y|x) =  \frac{P(x|y)P(y)}{P(x)} = \frac{(n+1)!}{n_1!n_0!}\theta^{n_1}(1-\theta)^{n_0} \]
Then if we let $x_i=1$ if the sun rose on the $i$th day according to our data, we want to find the probability of the the sun will rise tomorrow. Our data says that the sun rose for every previous day, then we get
\[ P(1|x) = \frac{P(x1)}{P(x)} = \frac{n_1 +1}{n+2} \] Which translates to
\[P(\text{sunrise tomorrow }|\text{ all previous sunrises})  =  \frac{\#\text{days the sun has risen} +1}{\#\text{days the sun has risen}+2} .\]
So according to Bayes' and Laplace's rule if the sun has risen for $10^{12}$ days (the approximate age of the earth) then the probability the sun will not rise tomorrow is $\frac{1}{10^{12}+2}$.

\subsection{Algorithmic Complexity and Solomonoff Induction}

This section will be split into two parts: the problem of algorithmic complexity and the problem of induction. For a complete analysis of these problems see \citep{ming2014kolmogorov}. In this section we will use the following notation. Let $\mathbb{B}^*$ denote the set of all finite binary strings, and let $\ell(p)$ denote the length of a binary string $p\in \mathbb{B}^*$.

\cite{kolmogorov1963tables} proposed a complexity measure 
of how complex a given string (or number) is. It is the length of the smallest program which when inputted into a (often Universal) Turing machine will produce the string. More formally, \begin{definition}
	Given a Turing machine $T$, the Kolmogorov complexity of a string $x$ is \[ K_T(x) = \min_p \{ \ell(p) \ :\ T(p) = x \} \] 
\end{definition} There has been much work done on Kolmogorov complexity, and its application to Artificial Intelligence, as well as to compression. For a  detailed description on these topics and the properties of Kolmogorov complexity, see \cite{Hutter:04uaibook}. Although Kolmogorov complexity is incomputable, it is limit computable. A function $f(x)$ is limit computable if there exists a computable function $\hat{f}(x,t)$ such that $f(x) = \lim_{t\to\infty} \hat{f}(x,t) $. This incomputabiliy comes from the fact that the programs being checked may not halt, and there is no way to determine which programs will not halt.

The problem of induction, that is, predicting or inferring what will come next given some past data, has been extensively studied. In 1964 Solomonoff proposed a solution to the induction problem \citep{solomonoff1964formal,solomonoff1964formal2}.

Solomonoff Induction is based on Epicurus' principle of multiple explanations and Occam's Razor. Together these principles state that one should always consider every theory (or model) which is consistent with the data, however give preference to more simple theories (or models). Solomonoff represents the concept of complexity of a model by the length of a program which `computes' the model.

 \cite{solomonoff1964formal,solomonoff1964formal2} proposed four equivalent solutions to the induction problem, then demonstrated these induction techniques on three different induction tasks. The three induction tasks were induction on a Bernoulli sequence, induction on sequences with symbol constraints, and phrase structure grammars in coding for induction. Since the original paper the definition has been refined, and more applications have been demonstrated.

A description of Solomonoff Induction is as follows.

\begin{definition}
	Given a universal monotone Turing Machine $U$, the universal (Solomonoff) a-priori probability semi-measure of $x\in \mathbb{B}^* $ is defined as\[ M(x) = \sum_{p\ :\ U(p) = x*} 2^{-\ell(p)} \] where $x*$ denotes a string which starts with $x$.
\end{definition}
The programs being summed over in the universal a-priori probability represent descriptions of $x$, and the $2^{-\ell(p)}$ is a weight on the complexity of the description. In this way we define complexity as description length. This is universal in the sense that every program is considered.

For conditional probability we have,
\begin{definition}
	The corresponding conditional probability of $x$ given $y$ is \[ M(x|y) = \frac{M(yx)}{M(y)}=\frac{\sum_{p:U(p) = yx*} 2^{-\ell(p)}}{\sum_{p:U(p) = y*} 2^{-\ell(p)}} \] where $yx$ denotes the concatenation of $y$ with $x$.
\end{definition}

As an example of this induction,

\begin{exmp}
	If 0 denotes rainy and 1 denotes sunny and we wish to calculate the probability it will be sunny given the past data 0101000, we calculate
	\[ M(1|0101000) = \frac{M(01010001)}{M(0101000)}=\frac{\sum_{p:U(p) = 01010001*} 2^{-\ell(p)}}{\sum_{p:U(p) = 0101000*} 2^{-\ell(p)}} \]
\end{exmp}

$M$ is a universal semi-measure \citep{ming2014kolmogorov} and, like Kolmogorov complexity, it is incomputable. The incomputability comes from the fact that one cannot compute every possible program on a universal Turing machine as some may not halt. 
Additionally it is required that every program be ran. Whereas for Kolmogorov complexity, one can show that only a finite amount of programs need to be ran. 
However, even though it it incomputable, it has been shown that $M$ is limit computable and there are algorithms which can approximate it. This will be discussed further in later sections.

The importance of $M$ comes from the following theorem proved by \cite{solomonoff1978complexity}, that is, Solomonoff induction.

\begin{theorem}[\cite{solomonoff1978complexity}]
\label{solthm}
	Let $\mu$ be a computable measure with $x_1,x_2,\ldots $ distributed according to $\mu$. Then the total squared error between $M$ and $\mu$ will be finite, specifically
	\[ \mathbb{E}_{\mu} \left( \sum_{t=1}^{\infty}  (M(x_{t+1}=1|x_1\ldots x_t ) - \mu (x_{t+1}=1|x_1\ldots x_t) )^2\right) \leq K(\mu) \frac{\ln (2)}{2} <\infty \]
\end{theorem}
where $K(\mu)$ denotes the Kolmogorov complexity of $\mu$. This theorem states that the expected value of the sum of the squares of the difference between $M$ and $\mu$ is bounded by a constant. Hence for any computable measure $\mu$, Solomonoff's prior will have a bounded error in prediction. Moreover it has been shown that any probability distribution $P$ that satisfies Theorem \ref{solthm} in place of $M$ must be incomputable \cite{solomonoff2003kolmogorov}. This shows that the incomputability of $M$ is not a bane but a requirement to be able to have bounded errors. In turn this demonstrates that any computable statistical procedure used in practice must have infinitely many errors on some sequences.

For a description of the application of Solomonoff Induction to Artificial Intelligence see \cite{solomonoff1985application}. Nearly forty years after the original paper, Solomonoff gave the Kolmogorov lecture \citep{solomonoff2003kolmogorov} after receiving the Kolmogorov medal, in which he discussed his Universal distribution and how it relates to AI. Solomonoff mentioned how the study of Artificial Intelligence had, for a long time, been completely absent of any notion of probability. However, to be able to solve the big problem, that is the creation of strong (general) Artificial Intelligence, one must inevitably use probability. For a description of how Solomonoff Induction can be used to create a theoretically optimal agent for reinforcement learning, that is a (theoretical) strong Artificial Intelligence, see \cite{Hutter:04uaibook} and Section \ref{aixisec}.

\subsection{Approximation schemes for Solomonoff Induction}
There are algorithms which can approximate Kolmogorov complexity and the Solomonoff prior, most notably Levin search \citep{levin1973universal}, Hutter search \citep{hutter2002fastest}, and the Optimal Order Problem solver  \citep{schmidhuber2004optimal}. 

Levin search \citep{levin1973universal} is an algorithm for solving a given inversion problem. Given a function $f$ and a value $y$, the Levin search algorithm inverts the function $f$. Levin search sets $i:=0$ and executes every input to the function, $x\in\mathbb{B}^*$, with $\ell (x) \leq i$ for time $2^i 2^{-\ell (x)}$ steps, then sets $i:=i+1$ and repeats, until an $x$ is found such that $f(x)=y$. 

Hutter Search \citep{hutter2002fastest} is a general speedup algorithm for any given problem. It works by performing three different tasks and sharing resources between them. The first task is to prove that other functions (programs in a universal Turing machine) are equivalent to the desired program and that these functions have time bounds, all of this in formal logic. The second task is to compute the time bound of every program which satisfies the condition of the first task, and this computation is split up in a similar way to Levin Search. The third task is to run the function (program) which has the best time bound from the second task. 

The G\"odel machine \citep{schmidhuber2007godel}
is a self-improving solver. It works by first solving (some) problems, then while solving new problems searching for other solvers that are able to out-perform itself, and changing its own solver to the new solver when it finds a superior one.

\subsection{The optimal agent AIXI}
\label{aixisec}
Reinforcement learning \citep{sutton1998reinforcement} is a paradigm in artificial intelligence where an agent is given observations and reward (a real number) by an environment, and the agent performs actions which has some effect on the environment. The goal is to maximise the future reward based on the history of interactions with the environment. A simple example of a reinforcement learning environment is an agent playing a game like tic-tac-toe, and the agent receives $+1$ reward for winning the game, $0$ for drawing the game, and $-1$ for losing the game. Below is a diagram of the agent environment (Env) interaction
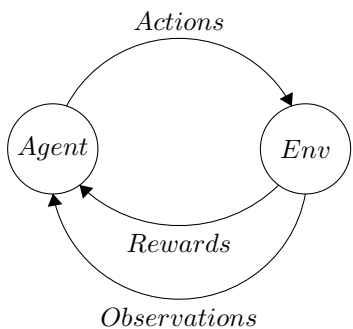
\begin{center}
\begin{tikzpicture}[scale=0.2]
\tikzstyle{every node}+=[inner sep=0pt]
\draw [black] (35.7,-27.6) circle (3);
\draw (35.7,-27.6) node {$Agent$};
\draw [black] (52.5,-27.6) circle (3);
\draw (52.5,-27.6) node {$Env$};
\draw [black] (36.61,-24.758) arc (152.10789:27.89211:8.475);
\fill [black] (51.59,-24.76) -- (51.66,-23.82) -- (50.77,-24.28);
\draw (44.1,-19.75) node [above] {$Actions$};
\draw [black] (50.722,-30.001) arc (-45.59346:-134.40654:9.464);
\fill [black] (37.48,-30) -- (37.7,-30.92) -- (38.4,-30.2);
\draw (44.1,-33.2) node [below] {$Rewards$};
\draw [black] (52.493,-30.585) arc (-10.21265:-169.78735:8.528);
\fill [black] (35.71,-30.58) -- (35.36,-31.46) -- (36.34,-31.28);
\draw (44.1,-38.1) node [below] {$Observations$};
\end{tikzpicture}
\end{center}
Solomonoff Induction and reinforcement learning can be used together to construct the optimal agent AIXI \citep{hutter2000theory}. The agent AIXI, described in \citep{Hutter:04uaibook}, is optimal in the sense that there does not exist another agent which performs better than AIXI in all possible environments.

Given the history $o_1r_1\ldots o_{k-1} r_{k-1}$ where $o_i$ is the observation at the $i$th time step, and $r_i$ is the reward at the $i$th time step, AIXI takes action $a_k$ defined by \begin{equation}
	a_k := \arg \max_{a_k} \sum_{o_k r_k}\ldots \max_{a_m} \sum_{o_m r_m} [r_k+\ldots r_m] \sum_{q:U(q,a_1\ldots a_m ) = o_1r_1\ldots o_m r_m} 2^{-\ell(q)}
\end{equation}
Here $q$ is a binary string, $U$ is a universal Turing machine, and $\ell$ is the length function. Much like Solomonoff Induction, AIXI is incomputable, as there is no way to determine if a given input to the universal Turing machine $U$ will halt. However, it is still possible to limit-compute and approximate AIXI, such as AIXItl: AIXI with a time bound, $t$, on the running time of $U$ and a length bound, $l$, on the programs $q$ \citep{Hutter:04uaibook}.

\section{Quantum Algorithms for Universal Prediction}
Although the exact Kolmogorov complexity and Solomonoff prior of a given $x\in \{0,1\}^*$ are incomputable, approximations are not. Finding efficient approximations of either would provide a powerful form of compression and prediction. One interesting approximation of the Solomonoff prior is the Speed prior \citep{schmidhuber2002speed}. The speed prior is very much like the Solomonoff prior, however it essentially accounts for running time of each program on a universal Turing machine, as opposed to running each program until it halts, as in the Solomonoff prior. 

The speed prior is not the only approximation. There is also universal search, however both of these still take time scaling exponentially in the length of the largest program. In this section, we give a quantum algorithm which takes advantage of the quantum counting algorithm to compute a fixed length speed prior, with a quadratic speedup compared to the classical method. We will also give a reasoning to why we suspect that the speed prior is $\# \textbf{P}$-complete, and as discussed in the previous section, that implies it is unlikely that there exists a quantum algorithm which has an exponential speedup over the classical method. Then we define a quasi-conditional Speed probability which can be efficiently computed with a quantum algorithm related to the conditional Speed prior algorithm. Lastly we present AIXIq a quantum computing approximation of AIXI and give some results about the potential speedup it possesses over classical methods.

\subsection{Speed Prior}
Although Solomonoff's prior has many desirable theoretical properties, when performing inference we often want a prior that is more easily computable, but still has strong theoretical properties. To this end we will use the Speed Prior \citep{schmidhuber2002speed}. The Speed prior is much like the Solomonoff prior, however it takes into account the amount of time each program has been running. This immediately gives a more practical prior than Solomonoff's prior. Before we define the speed prior itself we first need to define what we mean by amount of time each program has been running. For that we use the FAST algorithm and notation which uses it.
\begin{definition}[\cite{schmidhuber2002speed}]
	$FAST$ algorithm: For $i=1,2\ldots$ perform PHASE $i$. PHASE $i$: execute $\lfloor 2^{i-\ell(p)}\rfloor$ instructions from all program prefixes $p$ satisfying $\ell(p)\leq i$, and subsequently write the outputs on adjacent sections of the output tape, separated by blanks.
\end{definition}

\begin{definition}[\cite{schmidhuber2002speed}]
	Given program prefix $p$, write $p\to x$ if our TM reads $p$ and computes an output starting with $x\in\{ 0,1\}^*$, while no fixed prefix of $p$ consisting of less than $\ell(p)$ bits outputs $x$. Write $p\to_i x$ if $p\to x$ in PHASE $i$ of $FAST$.
\end{definition}
Much like Solomonoff's prior the complete speed prior is defined over all programs. Additionally it is defined over all program times.

\begin{definition}[\cite{schmidhuber2002speed}]
	The complete speed prior $S$ on $x\in\mathbb{B}^*$ is
	\[ S(x) = \sum_{i=1}^{\infty}2^{-i}S_i (x); \ \text{ where } \ S_i(\lambda )=1;\ S_i(x) = \sum_{p\to_i x}2^{-\ell(p)} \text{ for } x\succ \lambda  \]
\end{definition}

Having to execute every possible program for every amount of time is somewhat impractical, this leads us to our finitely computable fixed length speed prior which we will be using.

\begin{definition}
	The speed prior (which in this case is finitely computable and fixed length) $S$ on $x\in\mathbb{B}^*$ is
	\[ S(x) = \sum_{i=1}^{n^2}2^{-i}S_i (x); \ \text{ where } \ S_i(\lambda )=1;\ S_i(x) = \sum_{p\in\mathbb{B}^n:\ p\to_i x}2^{-\ell(p)} \text{ for } x\succ \lambda  \]
	Where $n=\ell(x)$
\end{definition} 

Here $n^2$ is chosen so that the universal Turing machine component of $S$ takes polynomial time (specifically a small polynomial). Unfortunately because of this restriction we have no way to guarantee the convergence to the speed prior, however we are not able to guarantee convergence if we run universal Turing machines for any finite time.

When performing inference a conditional probability is required, for this we defined the conditional speed prior.
\begin{definition}
	The conditional speed prior of $y\in\mathbb{B}^*$ given $x\in\mathbb{B}^*$ is defined as
	\begin{equation}
		S(y|x) = \frac{S(xy)}{S(x)}
	\end{equation}
\end{definition}

\subsection{Complexity of $S$}
$S$ is the problem of the fixed-length speed prior.

We will show that if there did exist a quantum algorithm (in $\mathbf{BQP}$) which could approximate $S$ exponentially faster then the classical method that would imply that $\mathbf{BQP} = \# \mathbf{P}$. We have been unable to prove that $S$ is $\# \mathbf{P}$-hard, however under this assumption this will immediately show it is in a superset of $\mathbf{BQP}$, since $\mathbf{BQP}\subseteq \mathbf{NP}\subseteq \#\mathbf{P}$.

\begin{lemma}
	$S\in \#\mathbf{P}$
\end{lemma}
\begin{proof}
	To compute the $S(x)$ we need to count the number of programs $p$ of a length $n$ such that $p\to_i x$ for $i$ from 1 to $n^2$. For each program $p$, determining if $p\to_i x$ for $i$ from 1 to $n^2$ can be done in polynomial time. The time taken to determine this for all $p\in\{0,1 \}^n$ grows exponentially in $n$, and is thus an $\mathbf{NP}$ problem. The class $\# \mathbf{P}$ is the counting class for $\mathbf{NP}$ problems, in particular it can count the number of programs satisfying this property, therefore
	we have that $S\in\# \mathbf{P}$.
\end{proof}

We suspect that given a $\#\mathbf{P}$-complete problem one can simulate the problem with the speed prior without a polynomial slowdown, however we have not yet been able to prove this claim.

\begin{conjecture}
		$ S$ is $\#\mathbf{P}$-hard  
\end{conjecture}

Together this would give us the $\#\mathbf{P}$-completeness of $S$.

\begin{conjecture}
	$S$ is $\#\mathbf{P}$-complete
\end{conjecture}

As a corollary of this, we have that if there did exist a quantum algorithm which could solve a $S$ in polynomial time, then the polynomial hierarchy would collapse, meaning having any number of $\mathbf{NP}$-complete oracles is equivalent to having no $\mathbf{NP}$-complete oracles, as mentioned in \cite{aaronson2011computational}.

\begin{corollary}
	Under the assumption that $S$ is $\#\mathbf{P}$-complete, if there exists a quantum algorithm which can compute $S$ in polynomial time then the polynomial hierarchy collapses.
\end{corollary}

\subsection{Classical Universal Prediction}
\label{classspdsec}
First we will give the classical algorithm to compute the speed prior.

To compute the speed prior of $x$ we can use the following Levin-search-style algorithm.

\begin{center}
\begin{algorithm}[H]
\SetAlgoLined
\label{classpdalg}
Given $x$\;
  $S:= 0$\;
  $n = \ell (x)$\;
   $i:=1$\;
  \eIf{$x=\lambda$}{$S:=1-2^{-n^2}$}{
 \While{$i\leq n^2$}{$num_i:=0$\; \For{$p \in \{0,1\}^n$}{\If{$p\to_i x$} {$num_i :=num_i + 1$}}$S := S+2^{-(i+n)}\cdot num_i$\;
 $i:=i+1$\;
 }}
 \KwResult{$S$ }
 \caption{Classical Computable fixed-length speed prior algorithm}
\end{algorithm}
\end{center}

The first loop requires $n^2$ iterations, each going for $1\leq i \leq n^2$ time, and the second loop requires $2^n$ iterations since the size of $\{0,1\}^n$ is $2^n$. Therefore the run time of this algorithm is $O(n^42^{n})$.

To compute $S(x|y)$, one needs only to compute $S(y)$ and $S(yx)$, then perform division of the two. Using the algorithm above, approximating $S(x|y)$ will require $O(n^42^{n})$ time, since we are just using it twice.

\subsection{Universal Prediction with Quantum Computing}

\subsubsection{Quadratic speedup for Speed Prior}
If the probability being measured in a quantum computation is exponentially small then it requires at least an exponential number of repeats to find the probability with sufficiently small error. Therefore we cannot, at least naively, use a quantum version of the above algorithm to gain an exponential speedup. 

It turns out, however, that for a lot of problems we can gain a quadratic speedup. This comes from the use of the quantum counting algorithm \citep{brassard1998quantum}; as mentioned in section 4, the quantum counting algorithm is a combination of Grover search and phase estimation which can find $S(x)=|\{ p\in \{0,1\}^{n}: p\to_{n^2} x\}|$ in $O(n^2 2^{{n}/2})$ time, opposed to the classical algorithm which takes $O(n^22^{n})$ time. Let $QCA(x)$ denote the result of the quantum counting algorithm being run on $n=\ell(x)$ qubits counting the set $S(x)$. The quantum counting algorithm is used in place of lines 9-14 of the classical \ref{classpdalg} used to compute the finitely computable speed prior.

\begin{center} 
\begin{algorithm}[H]
\SetAlgoLined
Given $x$\;
  $S:= 0$\;
  $n = \ell(x)$\;
  $i:=1$\;
  \eIf{$x=\lambda$}{$S:=1-2^{-n^2}$}{
 \While{$i\leq n^2$}{$num_i:=QCA(x)$\; $S := S+2^{-(i+n)}\cdot num_i$\;
 $i:=i+1$\;
 }} 
 \KwResult{$S$ }
 \caption{Quantum Counting speed prior algorithm}
\end{algorithm}	
\end{center}

Thus the quantum computation of the finitely computable speed prior
has running time $O(n^42^{n/2})$. This is shown in the calculation below:
\begin{align*}
	\sum_{i=1}^{n^2}  i2^{n/2}  &= O(n^4 2^{n/2}) 
\end{align*}
\begin{theorem}[Quantum Correctness]
	Given the an $x\in\{ 0,1\}^*$, and the Quantum Counting algorithm function $QCA:\{0,1\}^*   \to \mathbb{N}$ such that $\sin^2\left( \frac{\theta}{2}\right)=\frac{QCA(x)}{2^{n+1}}$, to find $\theta$ with $m$ bits of accuracy with probability $1-\epsilon$ the 
	quantum finitely computable fixed length speed prior algorithm requires $O(m+n+\log(2+\frac{n^2}{2}))$ registers and $O(n^42^{n/2})$ time.
\end{theorem}
\begin{proof}
This proof comes immediately from Theorem \ref{quancountcorr}. For a given  $x$ $QCA$ allows us to find the number of programs satisfying our property in $O(2^{n/2})$ steps. At each stage the universal Turing machine runs for $2^{i}$ time, so the total time of the algorithm is $\sum_{i=0}^{n^2} i 2^{n/2} =O(n^42^{n/2})$. The accuracy of each individual $\theta$ at a given stage is $m$ and is correct with probability $1-\epsilon$. This means that the probability of correctness of all the stages is $(1+\epsilon)^{n^2}$, thus for the probability of the complete algorithm being correct is $1+(1+\epsilon)^{n^2}-1$, therefore if we choose $\epsilon=\frac{1}{n^2}$ for each $QCA$ we will have sufficiently small final error.
\end{proof}
Although the time is still exponential, this quadratic reduction in time may still be useful. As mentioned earlier, to calculate the conditional finitely computable speed prior $S(x|y)$, one needs only to run the above algorithm twice.

\subsubsection{Exponential Speedup for Speed Prior}

Our second quantum algorithm to compute $S$ will be similar to the classical algorithm for computing $S(x)$, however we will use the modified Deutsch-Jozsa  algorithm as a subroutine to check every program $p\in\{0,1\}^{n}$.

Formally let \[f_{i}(p) = \begin{cases}
	1 &\text{if } p\to_i x \text{ for } p\in \{0,1\}^n \\
	0 &\text{otherwise}
\end{cases}\]

For $f_i$ 
 let $L$ be the number of elements mapped to 1, and $2^n-L$ be the number of elements mapped to 0. Recall from section \ref{DJsec} that the modified Deutsch-Jozsa algorithm is as follows, $f$ is no longer restricted to the condition that either $f(x)=0$ for all $x$ or $f(x)=0$ for exactly half of the inputs $x$.
 \begin{lemma}
	Using the Quantum Counting Algorithm for $f_i$ with $O\left( \frac{k}{\epsilon^2} \right)$ trials, the absolute error in estimating $L$ is at most $\epsilon$ with probability at most $2e^{-2k}$.
\end{lemma}
\begin{proof}
	This follows immediately from Lemma \ref{DJlem}.
\end{proof}
Now for the quantum algorithm we let $D_i$ denote the output of our modified Deutsch-Jozsa algorithm with the oracle function $f_{i}$, and let the $r$-average of a quantum algorithm be the empirical mean of the quantum algorithm with $r$ trials.

Then our quantum algorithm to compute the finitely computable speed prior is as follows:
 
\begin{center}
 \begin{algorithm}[H]
\label{qexpspdalg}
\SetAlgoLined
Given $x$\;
  $S:= 0$\;
  $n:=\ell (x)$\;
  $i:=1$\;
  \eIf{$x=\lambda$}{$S:=1-2^{-n^2}$}{
 \While{$i\leq n^2$}{$num_i:= r\text{-average of }D_i$\; $S := S+2^{-(i+n)}\cdot num_i$\;
 $i:=i+1$\;
 }}
 \KwResult{$S$ }
 \caption{Quantum Speed prior algorithm}
\end{algorithm}	
\end{center}

\begin{theorem}
\label{expspdthm}
	The quantum speed prior algorithm computes the speed prior with absolute error at most $O(\epsilon)$ with probability at most $2n^2e^{-2k}$ in time $O\left(\frac{k}{\epsilon^2}n^4\right)$
	\end{theorem}
\begin{proof}
The absolute error at each $num_i$ step (line 8) is at most $\epsilon$, therefore the absolute error in line $9$ is at most $2^{-(i+n)}\epsilon$, therefore the absolute error in the final $S$ is at most \[ \sum_{i=1}^{n^2} 2^{-(i+n)}\epsilon = (1-2^{-n^2})2^{-n}\epsilon = O(\epsilon), \] by Boole's inequalities the probability of $num_i$ having absolute error $\epsilon$ for all $i$ is bounded by $2n^2e^{-2k}$, since the probability of a single $num_i$ having absolute error $\epsilon$ is at most $2e^{-2k}$.
By lemma 6.4.1 we have that the number of trials required to achieve an absolute error of at most $\epsilon$ with probability at most $2e^{-2k}$ is $m=\frac{k}{\epsilon^2}$. Since the computation of each $D_i(n)$ takes $O(i)$ time, then each $num_i$ will take $O(i\frac{k}{\epsilon^2})$ time. Addition and multiplication take constant time. Therefore the total time taken is \[ \sum_{i=1}^{n^2} O\left( i\frac{k}{\epsilon^2}\right) = O\left( n^4\frac{k}{\epsilon^2}\right) \]
\end{proof}

Although this appears to be an exponential speedup over the classical method there is a problem. Our average output may be exponentially small, potentially of size $O(2^{-n})$. This would mean that we need $\epsilon$ to be at most $2^{-n}$, which means the time taken will be $O\left(k2^{2n}n^4\right)$ which is not a speedup on the classical methods which take time $O(n^42^n)$. This also means that for the conditional speed prior, which is the quotient of two speed priors, we not have a speedup over classical methods either.
\begin{exmp}
	Let $y$ be a random binary sequence, then to compute $S(1|y)$ we need to compute $S(y)$ and $S(y1)$. The $\epsilon$ required for $S(y)$ will be at most $2^{-n}$, where $n=\ell (y)$ since $y$ is a random sequence. Then by the previous theorem the time required will be $O\left(\frac{k}{(2^{-n})^2}n^4\right) =O(kn^4 2^{2n}) $.
\end{exmp}
This example shows that there are cases where \ref{qexpspdalg} has no speed improvement over the classical method. However we may be able to avoid the occurrence of small $\epsilon$'s.

\subsubsection{Quasi-conditional prediction} \label{quasicondsec}
As mentioned earlier, it is unlikely that we will be able to find a quantum algorithm to compute $S$ exponentially faster than our classical algorithm. This is because we may not suffer the problem of small probabilities. Since we cannot exactly take a quotient mid quantum computation then we need to use some structure of the problem to do something similar. This leads us to the quasi-conditional speed prior, which is exponentially faster than classical methods. Additionally, we would like our quasi-conditional speed prior to predict ``well'' in the sense that it does not differ too much from the conditional speed prior.

Formally, our quasi-conditional speed prior is defined as follows:

\begin{definition}
	The quasi-conditional finitely computable speed prior of $x\in \{0,1\}^*$ given $y\in \{0,1\}^*$ is \begin{equation}
		  S'(x,y) = \sum_{i=1}^{n^2}2^{-i}S_i' (x,y); \ \text{ where } \ S'_i(\lambda,y )=1;\ S'_i(x,y) = \sum_{p\in\mathbb{B}^n:\ py\to_i yx}2^{-\ell(p)} \text{ for } x\succ \lambda
	\end{equation}
	Where $n = \ell(x)$
\end{definition}

Our quantum algorithm to compute $S'$ will be similar essentially the same as the quantum algorithm used to compute $S$.

Formally let \[f_{i}'(p) = \begin{cases}
	1 &\text{if } py*\to_i yx \text{ for } p\in \{0,1\}^n \\
	0 &\text{otherwise}
\end{cases}\]

Just like before 
 let $L$ be the number of elements mapped to 1, and $2^n-L$ be the number of elements mapped to 0.

\begin{lemma}
	Using the Quantum Counting Algorithm for $f_i'$ with $O\left( \frac{k}{\epsilon^2} \right)$ trials, the absolute error in estimating $L$ is at most $\epsilon$ with probability at most $2e^{-2k}$.
\end{lemma}
\begin{proof}
	This follows immediately from Lemma \ref{DJlem}.
\end{proof}

Now for the quantum algorithm we let $D_i'$ denote the output of our modified Deutsch-Jozsa algorithm with the oracle function $f_{i}'$, and let the $r$-average of a Quantum algorithm be the empirical mean of the Quantum algorithm with $r$ trials.

Then our quantum algorithm to compute the quasi-conditional finitely computable speed prior is as follows:

\begin{center}
\begin{algorithm}[H]
\label{quasispdalg}
\SetAlgoLined
Given $x$\;
  $S':= 0$\;
  $n:=\ell (x)$\;
  $i:=1$\;
  \eIf{$x=\lambda$}{$S':=1-2^{-n}$}{
 \While{$i\leq n^2$}{$num_i:= r\text{-average of }D_i'$\; $S' := S'+2^{-(i+n)}\cdot num_i$\;
 $i:=i+1$\;
 }}
 \KwResult{$S'$ }
 \caption{Quantum Quasi-conditional speed prior algorithm}
\end{algorithm}
\end{center}

Then to show the computation time for our algorithm we have the following theorem. The proof almost exactly like that of the regular speed prior.
\begin{theorem} 
	The Quantum quasi-conditional speed prior algorithm computes the quasi speed prior, $S'$, with absolute error at most $O(\epsilon)$ with probability at most $2n^2e^{-2k}$ in time $O\left(\frac{k}{\epsilon^2}n^4\right)$
	\end{theorem}
	The proof of this theorem is identical to the proof of \ref{expspdthm}.

Thus we can compute the quasi-conditional speed prior exponentially faster than the classical conditional speed prior. However, we have not yet shown that the quasi-speed prior is a reasonable approximation of the classical conditional speed prior.

\begin{conjecture}
	The quasi-conditional speed prior approximates the conditional speed prior sufficiently well.
\end{conjecture}

The same example does not cause the quasi-speed prior algorithm to take an exponential amount of time since the $\epsilon$ is not going to be as small, this is because for $S'$ the program begins with the sequence $y$, so being a random sequence does not require the $\epsilon$ to be of size $2^{-n}$. This is not a proof that there does not exist a sequence, or there only exists a small set of sequences which cause extremely small $\epsilon$.

\subsection{AIXIq}
Recall that the agent AIXI can be written as taking action $a_k$ where
\begin{equation}
	a_k := \arg \max_{a_k} \sum_{o_k r_k}\ldots \max_{a_m} \sum_{o_m r_m} [r_k+\ldots r_m] \sum_{q:U(q,a_1\ldots a_m ) = o_1r_1\ldots o_m r_m} 2^{\ell(q)}
\end{equation}

Observe that one can modify AIXI to use the speed prior instead of Solomonoff induction. The resulting agent, AIXI-Spd, will take action $a_k$ where
\begin{equation} \label{aixispdeq}
	a_k := \arg \max_{a_k} \sum_{o_k r_k}\ldots \max_{a_m} \sum_{o_m r_m} [r_k+\ldots r_m] S(o_1r_1\ldots o_m r_m| a_1\ldots a_m)
\end{equation}
Using the classical algorithm the time taken to compute AIXI-Spd is
\begin{theorem}
	Computational time of compute AIXI-Spd using \ref{classpdalg} and search over observations and rewards is $O(|\mathcal{O}|^m|\mathcal{A}|^m (nm)^42^{nm})$. 
	
	Where $|\mathcal{O}|$ is the size of the set of observations, $|\mathcal{A}|$ is the size of the set of actions, and each element of $\mathcal{O},\mathcal{A}$ is bounded in size by $n$.
\end{theorem} 
\begin{proof}
	The time taken to compute $S(o_1r_1\ldots o_m r_m| a_1\ldots a_m)$ is $(nm)^42^{nm}$ when $o_i,a_i$ is bounded in size by $n$, shown in section \ref{classspdsec}. The number of possible observation reward pairs is $|\mathcal{O}|^m|\mathcal{A}|^m$, therefore the total time required to compute Equation \ref{aixispdeq} is $O((nm)^42^{nm} |\mathcal{O}|^m|\mathcal{A}|^m)$.
\end{proof}

Then from this we can define our AIXIq
\begin{definition}
	Given the history $o_1r_1\ldots o_{k-1} r_{k-1}$ where $o_i$ is the observation at the $i$th time step, and $r_i$ is the reward at the $i$th time step, AIXIq takes action $a_k$ defined by \begin{equation} \label{aixiqeq}
	a_k := \arg \max_{a_k} \sum_{o_k r_k}\ldots \max_{a_m} \sum_{o_m r_m} [r_k+\ldots r_m] S_q(o_1r_1\ldots o_m r_m| a_1\ldots a_m)
\end{equation}
Here $S_q$ is the quasi-conditional speed prior computing on a quantum computer using Algorithm \ref{quasispdalg}, with $\epsilon = \frac{1}{nm2^{|\mathcal{O}||\mathcal{A}|(m-k)}}$ and $k=100$.
\end{definition}

 The $\epsilon$ is chosen this way to prevent the cumulation of errors, as we need to consider the errors of the quasi-conditional speed prior with every combination of action and observation. Since we suspect that the quasi speed prior is a close to the speed prior, and the $\epsilon$ is chosen so that the action $a_k$ chosen is sufficiently close to the action chosen by AIXI-Spd, then it should follow that AIXIq is a good approximation of AIXI-Spd.

\begin{conjecture}
	AIXIq is a good approximation of AIXI-Spd
\end{conjecture}

Then the time taken for AIIXq is as follows

\begin{theorem}	
	The computational time taken for AIXIq is $O(|\mathcal{O}|^m|\mathcal{A}|^m (nm)^62^{2(|\mathcal{O}||\mathcal{A}|(m-k))} )$
	
	Where $|\mathcal{O}|$ is the size of the set of observations, $|\mathcal{A}|$ is the size of the set of actions, and each element of $\mathcal{O},\mathcal{A}$ is bounded in size by $n$.
\end{theorem}
\begin{proof}
	The time taken to compute $S_q(o_1r_1\ldots o_m r_m| a_1\ldots a_m)$ is $(nm)^6 2^{2(|\mathcal{O}||\mathcal{A}|(m-k))}$ when $o_i,a_i$ is bounded in size by $n$ and $\epsilon = 1/(nm2^{|\mathcal{O}||\mathcal{A}|(m-k)})$, shown in section \ref{quasicondsec}. The number of possible observation reward pairs is $|\mathcal{O}|^m|\mathcal{A}|^m$, therefore the total time required to compute Equation \ref{aixiqeq} is $O(|\mathcal{O}|^m|\mathcal{A}|^m (nm)^6 2^{2(|\mathcal{O}||\mathcal{A}|(m-k))})$.
\end{proof}

Thus using quantum computing we can compute AIXI-Spd exponentially faster in $n$, although there is an exponential increase with respect to $|\mathcal{O}|$, $|\mathcal{A}|$ and $m$.

\section{Conclusion}
Quantum computing is an expanding field which has provided opportunities to perform computing in a completely different way to classical computing method. This has resulted in the study of the theoretical and practical aspects of quantum computing, we focused on the theoretical aspects.

We have provided a description of quantum computing, including the current state of quantum algorithms, a summary of the more popular quantum algorithms, including those used to gain speedups for classically hard problems, such as Shor's algorithm. Additionally we included evidence to suggest why it is unlikely that quantum computing could solve counting problems (exponentially) faster than classical computing. Counting is a major part of approximating the Speed prior thus we need to be able to show that  if we had a quantum algorithm which could approximate the Speed prior exponentially faster than classical methods then the polynomial hierarchy would collapse.

Additionally we presented two quantum algorithms to provide speedups for the problem of approximating the Speed prior. Firstly a quadratic speedup based on the quantum counting algorithm, and secondly an exponential speedup for a quasi Speed prior, since this was not the exact Speed prior it did not contradict the evidence mentioned previously, nor cause the polynomial hierarchy to collapse. Together we can use these algorithms to create an approximation of the super intelligent agent AIXI which is able to perform learning faster than a classical version.

There are many avenues to achieve Artificial General Intelligence, an approximation of the theoretically optimal agent AIXI is just one. All current approximations of AIXI are not viable. Our use of quantum computing techniques to construct a quantum algorithm to provide speedup over a classical algorithm is a way in which an approximation may become viable. The viability does then depend on the existence of reasonably-sized general purpose quantum computers, which at the time of writing do not exist.

We hope that future research may be able to come up algorithms which have improved computation time as well as physical quantum computing devices in which our quantum algorithms could be implemented.

\bibliographystyle{apalike}
\bibliography{TechReport_complete}

\section{List of Notation}
\begin{align*}
	\mathbb{N} &: \text{The set of Natural Numbers } \\
	\mathbb{Z} &: \text{The set of Integer Numbers }  \\
	\mathbb{Q} &: \text{The set of Rations Numbers }  \\
	\mathbb{R} &: \text{The set of Real Numbers }  \\
	\mathbb{C} &: \text{The set of Complex Numbers }  \\
	\tilde{\mathbb{C}} &: \text{The set of polynomial-time computable Complex Numbers } \\
	QTM &:\text{Quantum Turing Machine} \\
	\ket{0} &: \text{The vector } \begin{pmatrix}
		1 \\ 0
	\end{pmatrix} \\
	\ket{1} &: \text{The vector } \begin{pmatrix}
		0 \\ 1
	\end{pmatrix} \\
	\ket{a}\otimes\ket{b} &: \text{The tensor product of } \ket{a} \text{ and } \ket{b} \\
	a^{\dagger} &: \text{The conjugate transpose of } a \\
	H &: \text{The Hadamard gate} \\
	N &: \text{A natural number} \\
	\mathcal{QFT} &: \text{The Quantum Fourier Transform} \\
		\mathcal{QFT}^{-1} &: \text{The inverse Quantum Fourier Transform} \\
	\theta &: \text{An angle} \\
	\kappa &: \text{Condition number of a matrix} \\
	G &: \text{Grover iteration} \\
	\mathbf{P},\mathbf{NP},\mathbf{PP},\mathbf{BPP},\mathbf{EQP},\mathbf{BQP} &: \text{The complexity classes of the same names} \\
	perm &: \text{Permanent of a matrix} \\
	det &: \text{Determinant of a matrix} \\
	S &: \text{Speed prior} \\
	S' &: \text{Quasi-conditional Speed prior} \\
	K_T(x) &: \text{Kolmogorov complexity with Turing Machine } T \\
	\mathbb{E} &: \text{Expected value} 
\end{align*}

\begin{align*}
\ell &:\text{Length} \\
\mathbb{B} &: \text{The set } \{0,1\} \\
	\mathbb{B}^* &: \text{The set of all finite binary strings} \\
	\mathcal{P} &: \text{The power set} \\
	\mathcal{O} &: \text{The set of observations} \\
	\mathcal{A} &: \text{The set of actions} \\
	||\cdot ||_{tr} &: \text{The trace norm} \\
	Q &: \text{Set of Turing Machine states} \\
	L,R &: \text{Left and Right} \\
	\oplus &:\text{Addition modulo 2} \\
	\epsilon &: \text{Small positive real number usually denoting error} \\
	\delta &: \text{A real number, or the transition function} \\
	\mu &: \text{A computable measure} \\
	\phi &:\text{The phase} \\
	I &: \text{The identity matrix} \\
	c &: \text{Constant} \\
	M &:\text{Solomonoff Prior or a natural number} \\
	m &:\text{Number of bits of accuracy or the maximal time for AIXI} \\
	n &:\text{A number} \\
	k &:\text{A number} \\
	T &:\text{A number, usually denoting time} \\
	A &: \text{A matrix} \\
	L &: \text{A positive integer} \\
	x &:\text{A number or binary string} \\
	y &:\text{A number or binary string} \\
	z &:\text{A number or binary string} \\
	\rho &:\text{A number or binary string} \\
\end{align*}

\end{document}